\documentclass[acmsmall,nonacm,runningheads]{acmart}

\usepackage[linesnumbered,ruled,vlined]{algorithm2e}
\usepackage{enumitem}
\usepackage{float}
\usepackage[inkscapearea=page]{svg}
\usepackage{adjustbox}
\usepackage{thmtools,thm-restate}
\newtheorem{theorem}{Theorem}[section]

\usepackage{xcolor}
\usepackage{caption}

\usepackage{savesym}
\usepackage{subcaption}

\savesymbol{If}
\savesymbol{Else}
\savesymbol{For}
\savesymbol{EndIf}
\savesymbol{EndFor}

\savesymbol{Repeat}
\savesymbol{Until}

\usepackage{amsmath}
\usepackage[noend]{algpseudocode}
\definecolor{codegreen}{rgb}{0,0.6,0}
\definecolor{codegray}{rgb}{0.5,0.5,0.5}
\definecolor{codepurple}{rgb}{0.58,0,0.82}
\definecolor{backcolour}{rgb}{0.95,0.95,0.92}

\usepackage[many]{tcolorbox}
\tcbuselibrary{listings}

\usepackage{lipsum}

\lstdefinestyle{mystyle}{
    backgroundcolor=\color{backcolour},
    commentstyle=\color{codegreen},
    keywordstyle=\color{magenta},
    numberstyle=\tiny\color{codegray},
    stringstyle=\color{codepurple},
    basicstyle=\ttfamily\footnotesize,
    breakatwhitespace=false,
    breaklines=true,
    captionpos=b,
    keepspaces=true,
    numbers=right,
    numbersep=5pt,
    showspaces=false,
    showstringspaces=false,
    showtabs=false,
    tabsize=2
}
\newenvironment{descriptionFixed}[1][3cm]{
\begin{description}[labelwidth=#1,labelindent=0pt,leftmargin=0.5cm,align=left,itemsep=0pt]
}{
\end{description}
}

\newif\ifcomment
\commenttrue    % Uncomment for blue/red comments.
\newcommand{\junk}[1]{}

\newif\ifconf
\conftrue
\ifconf\else
\hideLIPIcs  %uncomment to remove references to LIPIcs series (logo, DOI, ...), e.g. when preparing a pre-final version to be uploaded to arXiv or another public repository
\nolinenumbers
\fi

\usepackage{xparse}	%Matthias
\newcommand{\varIndex}{\ensuremath{l}}

\newcommand{\rvX}[1][\varIndex]{\ensuremath{X_{#1}}}
\newcommand{\rvY}[1][\varIndex]{\ensuremath{Y_{#1}}}
\newcommand{\rvZ}[1][\varIndex]{\ensuremath{Z_{#1}}}
\newcommand{\rvXSum}{\ensuremath{X}}
\newcommand{\rvYSum}{\ensuremath{Y}}
\newcommand{\expX}[1][\varIndex]{\ensuremath{\mu_{#1}}}
\newcommand{\expXtilde}[1][\varIndex]{\ensuremath{\tilde{\mu}_{#1}}}
\newcommand{\exptilde}{\ensuremath{\tilde{\mu}}}
\newcommand{\expXhat}[1][\varIndex]{\ensuremath{\hat{\mu}_{#1}}}
\newcommand{\exphat}{\ensuremath{\hat{\mu}}}

\newcommand{\varAll}[1][\varIndex]{\ensuremath{\sigma^2}}

\newcommand{\numberVars}{\ensuremath{k}}

\newcommand{\varEnum}{\ensuremath{\varIndex \in [\numberVars]}}

\newcommand{\sumEnum}[1][EMPTY]{\ensuremath{\sum_{\varEnum} #1}}
\RenewDocumentCommand{\sumEnum}{O{Empty}}{\ensuremath{\sum_{\varEnum} #1}}

\newcommand{\prodEnum}[1][EMPTY]{\ensuremath{\prod_{\varEnum} #1}}
\RenewDocumentCommand{\prodEnum}{O{Empty}}{\ensuremath{\prod_{\varEnum} #1}}

\newcommand{\flowonuvbyi}{\ensuremath{f_{i,(u,v)}}}
\newcommand{\capofuv}{\ensuremath{c_{(u,v)}}}

\newcommand{\Expec}[1][EMPTY]{\ensuremath{\mathsf{Ex}\left(#1\right)}}
\RenewDocumentCommand{\Expec}{O{EMPTY}}{\ensuremath{\mathsf{Ex}\left(#1\right)}}

\newcommand{\cExpec}[1][EMPTY]{\ensuremath{\mathsf{Ex}\left(#1\right)}}
\RenewDocumentCommand{\cExpec}{O{EMPTY} O{EMPTY}}{\ensuremath{\mathsf{Ex}_{#1}\left(#2\right)}}

\newcommand{\Prob}[1][EMPTY]{\ensuremath{\mathsf{Pr}\left(#1\right)}}
\RenewDocumentCommand{\Prob}{O{EMPTY}}{\ensuremath{\mathsf{Pr}\left(#1\right)}}

\newcommand{\Varia}[1][EMPTY]{\ensuremath{\mathsf{Var}\left(#1\right)}}
\RenewDocumentCommand{\Varia}{O{EMPTY}}{\ensuremath{\mathsf{Var}\left(#1\right)}}

\newcommand{\expF}[1][EMPTY]{\ensuremath{\exp\left(#1\right)}}
\RenewDocumentCommand{\expF}{O{EMPTY}}{\ensuremath{\exp\left(#1\right)}}
\newcommand{\expT}[1][EMPTY]{\ensuremath{e^{#1}}}
\RenewDocumentCommand{\expT}{O{EMPTY}}{\ensuremath{e^{#1}}}

\newcommand{\param}{\ensuremath{\theta}}
\newcommand{\paramalpha}{\ensuremath{\theta_{\alpha}}}
\newcommand{\parambeta}{\ensuremath{\theta_{\beta}}}

\makeatletter
\newcommand{\pushright}[1]{\ifmeasuring@#1\else\omit\hfill$\displaystyle#1$\fi\ignorespaces}
\newcommand{\pushleft}[1]{\ifmeasuring@#1\else\omit$\displaystyle#1$\hfill\fi\ignorespaces}
\newcommand{\specialcell}[1]{\ifmeasuring@#1\else\omit$\displaystyle#1$\ignorespaces\fi}
\makeatother

\newcommand{\comm}[1][EMPTY]{ &{\qquad\qquad~\hfill\textnormal{[#1]}}}

\makeatletter
\newcommand{\nosemic}{\renewcommand{\@endalgocfline}{\relax}}%
\newcommand{\dosemic}{\renewcommand{\@endalgocfline}{\algocf@endline}}%
\let\oldnl\nl%
\newcommand{\nonl}{\renewcommand{\nl}{\let\nl\oldnl}}%
\makeatother

\makeatletter
\newcommand{\removelatexerror}{\let\@latex@error\@gobble}
\makeatother

 \newcommand{\tagIt}[1]{\refstepcounter{equation}\textnormal{({\theequation})} \label{#1}}

\newcommand{\SET}{\textbf{set~}}

\newcommand{\LET}{\textbf{let~}}

\newcommand{\COMPUTE}{\textbf{compute~}}

\newcommand{\UPDATE}{\textbf{update~}}

\newcommand{\eps}{\epsilon}
\AtBeginDocument{%
  \providecommand\BibTeX{{%
    \normalfont B\kern-0.5em{\scshape i\kern-0.25em b}\kern-0.8em\TeX}}}

\setcopyright{acmcopyright}
\copyrightyear{2021}
\acmYear{2021}

\acmMonth{8}
\acmArticleSeq{11}

\begin{document}
\title[Improved Multicommodity Flow Throughput]
{Improved Throughput for All-or-Nothing Multicommodity Flows with Arbitrary Demands}

%\title[Improved Bi-criteria Approximation for ANF]{Improved Bi-criteria Approximation for the All-or-Nothing Multicommodity Flow Problem in Arbitrary Networks}

%\title[Throughput Optimization in Arbitrary Networks]
%{Improved Throughput Optimization for the All-or-Nothing %Multicommodity Flows in Arbitrary Networks}

%\titlerunning{All-or-Nothing Multicommodity Flow Approximation}

\author{Anya Chaturvedi\mbox{$^\dagger$}}
\affiliation{\institution{School of Computing and Augmented Intelligence, Arizona State University, US}}
\email{anya.chaturvedi@asu.edu}
\author{Chandra Chekuri}
\affiliation{\institution{Dept. of Computer Science, University of Illinois at Urbana-Champaign, US}}
\email{chekuri@illinois.edu}
\authornote{[Research supported in part by NSF grant CCF-1910149.]}

\author{Mengxue Liu}
\affiliation{\institution{School of Computing and Augmented Intelligence, Arizona State University, US}}
\email{mengxueliu.hust231@gmail.com}
\author{Andr\'ea W.\ Richa}
\affiliation{\institution{School of Computing and Augmented Intelligence, Arizona State University, US}}
\email{aricha@asu.edu}

\authornote{[Research supported in part by
NSF CCF-1637393 and CCF-1733680, and DoD-ARO MURI
No.W911NF-19-1-0233 awards.]}

\author{Matthias Rost}
\affiliation{\institution{Faculty of Electrical Engineering and Computer Science, TU Berlin, Germany}}
%Department of Telecommunication Systems, Technische Universit\"at Berlin, Germany }
\email{mrost@inet.tu-berlin.de}

\author{Stefan Schmid}
\affiliation{\institution{Faculty IV, TU Berlin, Germany} }
\email{stefan.schmid@tu-berlin.de}

\authornote{[Research supported by the ERC Consolidator project AdjustNet, grant agreement No.~864228.]}

\author{Jamison W. Weber\mbox{$^\dagger$}}
\affiliation{\institution{School of Computing and Augmented Intelligence, Arizona State University, US}}
\email{jwweber@asu.edu}

%\Copyright{Anya Chaturvedi, Andr\'ea W. Richa, Matthias Rost, Stefan Schmid, and Jamison Weber}


\newcommand{\stefan}[1]{\textbf{stefan: #1}}
\newcommand{\chandra}[1]{{\color{blue}Chandra: #1}}

\begin{abstract}
Throughput is a main performance objective in communication networks. 
This paper considers a fundamental maximum throughput routing problem ---
the \emph{all-or-nothing multicommodity flow (ANF)} problem --- in arbitrary
  \emph{directed} graphs and in the practically relevant but challenging setting where {\em demands can
  be (much) larger than the edge capacities}. 
 Hence, in addition to assigning requests to valid flows for each routed commodity, an admission control mechanism is required, 
 which prevents overloading the network when routing commodities. 

Formally, the input for the ANF problem is an edge-capacitated
  directed graph $G=(V,E)$, where $n=|V|,$ and $k$ source-destination node-pairs
  $(s_i,t_i)$ of demand $d_i> 0$ and weight $w_i > 0$. The
  goal is to route a maximum weight subset of the given pairs (i.e., the weighted {\em throughput}), respecting the edge capacities; a pair
  $(s_i,t_i)$ is routed if all of its demand $d_i$ is routed from
  $s_i$ to $t_i$ (this is the all-or-nothing aspect);  splitting (fractional)
  flows is allowed. 
  
  We make several contributions.
  On the theoretical side we %obtain the first
  %substantially improved 
  present a bi-criteria
  approximation randomized rounding framework for this NP-hard problem that achieves a constant approximation of the throughput while only violating the edge capacities by a logarithmic factor. We present two
  non-trivial linear programming relaxations that can be used in the framework.
  %and show how to convert
  %their fractional %solutions into integer %solutions via randomized %rounding. 
  One is an
  exponential-size formulation (solvable in polynomial time using a separation oracle)
  that considers a ``packing'' view and
  allows a more flexible approach, while the other is a compact (polynomial-size) edge-flow formulation 
  %LP
  %formulation of Liu et al.\ (INFOCOM'19) 
  that allows for easy solving via standard LP solvers.
  We prove the non-trivial "equivalence" of the
  two relaxations and
  highlight the advantages of each of the two approaches. Via these, we obtain a polynomial-time randomized 
  algorithm that yields an arbitrarily good approximation on the weighted throughput
  while violating the edge capacity constraints by at most
  an $O(\min\{k,\log n/\log \log n\})$ multiplicative factor. 
  %This improves on the best-known previous result by Liu et al., which achieved a 1/3 throughput approximation and an edge capacity violation ratio
  %of $O(\sqrt{k \log n})$.
  We also describe a deterministic rounding algorithm by derandomization, 
  using the method of pessimistic estimators.

  We complement our theoretical results with a proof of concept empirical evaluation,
  considering a variety of network scenarios. 
  We study two different ways to solve the LP efficiently
   in terms of time and space:
  $(a)$ by solving the compact ANF formulation directly using an
  off-the-shelf solver, and $(b)$ by approximately solving the packing LP relaxation via a well-known
  multiplicative weight update (MWU) approach (based on Lagrangean relaxation) or via a faster MWU-based heuristic called permutation routing. We highlight the benefits of the ANF packing LP formulation by presenting some more general scenarios of interest to networking applications (such as routing along short paths or a small number of paths) that this formulation allows.
%for which our approach can be used. 
\end{abstract}

\maketitle
\renewcommand{\shortauthors}{A.\ Chaturvedi, C.\ Chekuri, M.\ Liu, A.\ W.\ Richa, M.\ Rost, S.\ Schmid, and J.\ W.\ Weber}

%\newpage

\section{Introduction} \label{sec:intro}

The study of routing and multicommodity flow problems is motivated by
many real-world applications, e.g., related to the optimization of
communication and traffic networks, as well as by the crucial role
flows and cuts play in combinatorial optimization~\cite{ChekuriE15}.
In this paper, we are interested in throughput optimization in the
context of communication networks serving multiple commodities.
Throughput is a most fundamental performance metric in many networks
\cite{mogul2012we}, and we are particularly interested in the
practically relevant scenario where flows have certain minimal
performance or quality-of-service requirements, in the sense that they
need to be served in an \emph{all-or-nothing} manner with respect to their respective demands.

Our problem belongs to the family of {\em all-or-nothing (splittable) multicommodity flow} %(ANF)}. 
problems.
%, however, i
In contrast to most existing literature, we consider a more realistic 
model in the following respects:
\begin{itemize}
    \item The underlying communication graph can be \emph{directed}.
      This is motivated by the fact that in most practical
      communication networks (e.g., optical networks or wireless
      networks), the available capacities in the different link
      directions can differ.

    \item A single commodity demand can be larger than the capacity of any single
      link or path.  Consider for example a bulk transfer, or the fact
      that traffic patterns are often highly skewed, with a small
      number of elephant flows consuming a significant amount of
      bandwidth resources~\cite{roy2015inside}.  Only {\em splittable}
      flows can serve such demands.
    
    \item The total demand can be larger than the network capacity. To
      make efficient use of the given network resources, we hence need
      a clever {\em admission control} mechanism, in addition to a
      routing algorithm.
\end{itemize} 

We define the {\em All-or-Nothing (Splittable) Multicommodity Flow
  (ANF)} problem as follows: It takes as input
%More formally, we model 
a flow network modeled as a capacitated directed graph $G(V,E)$, where
$V$ is the set of nodes, $E$ is the set of edges, and each edge $e$
has a capacity $c_{e}>0$. Let $n=|V|$ and $m=|E|$.  We are given a set
of source-destination pairs $(s_i, t_i)$, where $s_i,t_i\in V$, $ i
\in [k]$\footnote{Let $[x]$ denote the set $\{1,\ldots ,x\}$, for any
positive integer $x$.}, each with a given (non-uniform) demand $d_i>0$
and weight $w_i>0$.  The edge capacities $c_{e}$, the demands $d_i$
and the weights $w_i$ can be arbitrary positive functions on $n$ and
$k$, for any $e\in E$ or $i\in [k]$.  A valid set of flows for
commodities $1,\ldots,k$ in $G$ (i.e., a valid {\em multicommodity
  flow}),
	%from  node $s_i$ to node $t_i$,
	must satisfy standard flow conservation constraints for each
        commodity~$i$, which imply that the amount of flow for
        commodity $i$ entering a node $v$ has to be equal to the flow
        for commodity $i$ leaving $v$, if $v\not= s_i,t_i$.
	%(see Constraints (1-2) of the Integer Program in Formulation~\ref{form:compact});
	The {\em load} of an edge $e$, given by the sum of the flows
        for all commodities on $e$, must not exceed the edge's
        capacity $c_{e}$.
	%The flow can be {\em split} along many branching routes, provided that flow conservation and edge capacity constraints are satisfied.
	Commodity $i$ is {\em satisfied} if $d_i$ units of flow of
        this commodity can be successfully routed from $s_i$ to $t_i$  in the network.
        (See also our mixed integer program edge-flow formulation in Figure~\ref{fig:compact}).

        %, in
        %Fig.~\ref{form:compact}.)

	%Constraints (2) of the Integer Program in
	%Formulation~\ref{IP} represent flow %conservation constraints
	%for each commodity $i$ (those constraints are %normalized by
	%$d_i$ and hence represent the total fraction of $d_i$ being
	%%routed for commodity $i$ through node $v$), and Constraint
	%(1) gives the total %fraction of flow for commodity for
	%commodity $i$ leaving $s_i$. Constraint (3) %represents the
	%edge capacity constraints.

We aim to maximize the total profit of a subset of %number of
commodities that can be concurrently satisfied in a valid multicommodity flow. Specifically, 
%we consider a weighted generalization of this problem, where, in addition to the demands $d_i$, each commodity is given a weight $w_i$ and 
the goal is to find a subset $K'\subseteq [k]$ of commodities to be
concurrently satisfied such that the (weighted) {\em throughput},
given by $\sum_{i\in K'} w_i$, is maximized over all possible
$K'$. The flow can be {\em split} arbitrarily along many branching
routes (subject to flow conservation and edge capacity constraints)
and does not have to be integral.

The ANF problem was introduced in \cite{ChekuriKS13} as a relaxation
of the classical Maximum Edge-Disjoint Paths problem (MEDP) and is
known to be NP-Hard and APX-hard even in the restricted setting of
unit demands and when the underlying graph is a tree
\cite{ChekuriKS13,1997}. In directed graphs, even with unit demands,
the problem is hard to approximate to within an $n^{\Omega(1/c)}$
factor even when edge capacities are allowed to be violated by a
factor $c$~\cite{chuzhoy2007hardness}. When demands can exceed the
minimum capacity, strong lower bounds exist even in very restricted
settings \cite{ShepherdV}. Hence, the literature has followed a
bi-criteria optimization approach where edge capacities can be violated slightly. Namely, in this paper we seek an {\em
  $(\alpha,\beta)$-approximation algorithm}:
%for the ANF problem,
For parameters $\alpha \in (0,1]$ and $\beta \geq 1$, we seek a
  polynomial-time algorithm that outputs a solution to the ANF problem
  %multicommodity flow solution 
  %satisfying flow conservation constraints for each %commodity
  %$i$, 
  whose throughput is at least an $\alpha$ {\em fraction of the
    maximum throughput} and whose {\em load on any edge $e$ is at most
    $\beta$ times the edge capacity $c_{e}$}, with high
  probability.\footnote{ With probability at least $1-1/n^c$, where
  $c>0$ is a constant.} The parameter $\beta$ hence provides an upper
  bound on the {\em edge capacity violation ratio} (or {\em
    congestion}) incurred by the algorithm.

%\TODO{AR04/23: Hard to approximate within any factor? Within any constant factor? Removed the mention of ANF being NP-hard, since that already follows from NP-hard to approximate, etc. \bluecomment{JW04/23: The results from [1] are general results Chekuri refers to about a special class of problems to which ANF belongs. Their results state that no problem of this special class has a PTAS unless P=NP, so I take this to mean hard to approximate within any constant factor.} }
%\redcomment{AR: Is it APX-hard?
%What exactly are the inapproximability results for the ANF problem? We need to add citations. Please see my
%other TODO comment under related work.}\\ \redcomment{Jamison and Anya: Resolved through discussion with Stefan. Cited the %necessary papers.},

\subsection{Our Contributions}
	
This paper revisits a fundamental maximum throughput routing problem, 
the all-or-nothing multicommodity flow (ANF) problem, considering 
a more general and practical setting where the network topology can be 
an {\em arbitrary directed graph}, with {\em arbitrary, non-uniform commodity demands} that can be much larger than the edge capacities, in contrast to most of the existing work in the literature. 
This model is challenging as it not only requires a clever algorithm
to efficiently route the splittable commodities across the directed
and capacitated network, but also an admission control policy.

We make several contributions.  On the  theoretical side, we present  a bi-criteria
  approximation randomized rounding framework for this NP-hard problem that achieves a constant approximation of the throughput while only violating the edge capacities by a logarithmic factor.
%obtain substantially %improved bi-criteria %approximation algorithms for
%this NP-hard problem. 
More specifically,
  \begin{itemize}
  \item We present {\em two non-trivial ANF linear programming relaxations}: One is an
    exponential-size formulation (solvable in polynomial time using a
    separation oracle) that considers a {\em ``packing''} view and
    allows a more flexible approach, while the other is a {\em strengthened relaxation} of a
    %generalization of the
    {\em compact edge-flow} mixed integer program (MIP) formulation 
    %of Liu et al.~\cite{infocom19} that allows for {\em arbitrary non-uniform demands and weights} and
    that 
    %also 
    allows for easy solving via
    standard LP solvers.  We prove the non-trivial {\em "equivalence"} of the two
    relaxations and highlight the advantages of each of the two
    approaches.
  \item Via these relaxations,
  %We show how to convert the fractional solutions of the LP relaxations into integer solutions via 
  we obtain a polynomial-time {\em randomized rounding} algorithm that yields an {\em $(1-\epsilon)$ throughput
      approximation}, for any $1/m \leq \epsilon<1$, with an {\em edge
      capacity violation ratio} (also referred to as {\em congestion}) of $O(\min\{k,\log n/\log \log n\})$,
    with high probability. 
  
  \item We also present a {\em deterministic rounding algorithm by
    derandomization}, using the method of pessimistic
    estimators. Contrary to most algorithms obtained this way, our
    derandomized algorithm is simple enough to be also of relevance in
    practice.
\end{itemize}
 % We complement our theoretical results with an {\bf empirical evaluation},
  %considering a variety of network scenarios. 
  %Since solving the LP relaxation is the main
  %bottleneck in our algorithms, we study two different ways 
  %to solve it efficiently in terms of {\em time and space}:
  %$(a)$ by approximately solving the LP relaxation via a well-known
  %{\em multiplicative weight update} approach (based on Lagrangean %relaxation),
  %and $(b)$ by {\em solving the compact ANF formulation directly} %using an
  %off-the-shelf solver. We conclude by presenting some
  %more general scenarios for which our approach can be used. 
In addition, our packing framework for ANF has
interesting networking applications, beyond the specific model considered in this paper.
We discuss different examples, related to {\em unsplittable flows}, 
flows that are {\em split into a small number of paths},
{\em routing along disjoint paths} for fault-tolerance, using
{\em few edges for the flow}, or routing flow along  {\em short paths}.

%\redcomment{AR: 
%Will need to change narrative for experimental results to justify having MWU/permutation routing so that we can solve packing formulation directly and hence also address extensions if desired. change abstract, etc. accordingly. Will finalize this today after talking to Anya.\\
%On another note, I can add section numbers to the different results mentioned in this section, if this seems like a good idea (I am not sure).\\
%Lastly, maybe remove boldface for "theoretical contributions" and "engineer ..."
%}

As a proof of concept,
%On the {\bf practical side}, 
we show how to {\em engineer our algorithms} for
practical scenarios. To this end, 
%we identify the computation of relaxed LP solutions as the main
%bottleneck in our approach, and 
we couple three algorithms that allow one to compute 
the relaxed LP solutions efficiently, in terms of time and space, with both our randomized and derandomized algorithms. 
The first algorithm {\em directly solves the compact ANF formulation} using an
  off-the-shelf solver, in our case CPLEX; the second algorithm approximately solves the packing LP relaxation via a well-known
  {\em multiplicative weight update (MWU)} approach, based on Lagrangean relaxation; the last and third algorithm is a faster MWU-based heuristic called {\em permutation routing}. 
%Using our implementations, 
We provide general guidelines about the relative
efficacy of these algorithms in specific real-world networks.
As a contribution to the research community, to ensure reproducibility
and facilitate follow-up work, we will release our implementation (source code)
and experimental artefacts %together 
with this paper.

\subsection{Novelty and Related Work}
	\label{sec:related}
   We presented preliminary results leading to this journal article 
   at two conferences, at IEEE INFOCOM 2019~\cite{Liu19}
   and at PERFORMANCE 2021~\cite{perf21anf}. 
   In particular, our compact edge-flow LP formulation first appeared in~\cite{Liu19} and led to $(1/3,O(\sqrt{k\log n})
	)$-approximation guarantees for the
	ANF problem for the case of {\em uniform} demands and weights
        in directed graphs %, where $k$ is the number of commodities
        (note that while~\cite{Liu19} considered the case of uniform demands, there was also no restriction on how large these
        demands can be when compared to the edge capacities).  
        In this article, we significantly improve and generalize the
        randomized rounding framework outlined in~\cite{Liu19}, in
        several ways: $(a)$ We are able to achieve an arbitrarily good
        throughput approximation bound; $(b)$ our bound on the edge
        capacity violation ratio does not depend on the number of
        commodities $k,$\footnote{Unless $k$ is very small, $o(\log n/\log\log n)$ in which
        case we get an approximation bound of $k$.} and significantly
        improves on the bound of $O(\sqrt{k\log
          n})$in~\cite{Liu19}; and $(c)$ we were also able to
        accommodate arbitrary non-uniform demands and commodity
        weights.
	In addition, we provide a derandomized algorithm for the ANF problem and a more flexible packing MIP formulation for the ANF problem that leads to several interesting extensions of practical interest. Some of these ideas were sketched in our previous short paper~\cite{perf21anf},
	however, without any technical details, proofs or evaluations. 

	Other work on bi-criteria $(\alpha,\beta)$-approximation
        schemes for the ANF problem that are closely related to ours
        aims at keeping $\beta$ constant, while letting $\alpha$ be a
        function of $n$. The work of Chekuri et
        al.~\cite{ChekuriKS13,anf-journal,ChekuriKS05} %\redcomment{AR: Chandra, can you please make sure we are citing the correct papers here? If possible, please also doublecheck all of the related work bounds and citations below...} \chandra{Yes, this is the most relevant literature for ANF. We can always add additional refs etc if referees complain.} 
        is the most
        relevant and was also the first to formalize the ANF
        problem. Their work implies an approximation algorithm for the
        general (weighted, non-uniform demands) ANF problem in {\em
          undirected graphs} with $\alpha = \Omega(1/\log^3 k)$ and
        $\beta = 1$. A requirement of their algorithm is that $\max_i
        d_i\leq \min_e c_{e}$.  This is a strong assumption, since it eliminates all
        (undirected) networks $G$ where the above assumption fails,
        such as for example complete graphs with unit edge capacities
        and demands $2\leq d_i\leq n-1$, for all $i$.
	Hence, besides the fact that our approximation guarantees
        differ from those of~\cite{ChekuriKS13} (we have constant
        $\alpha$ and logarithmic $\beta$, while they achieve constant
        $\beta$ at the expense of a polylogarithmic $1/\alpha$), our
        results also apply to {\em any directed graph} $G$, without
        any assumptions on how $d_i$ compares to individual edge
        capacities. We note that even in undirected graphs and
        unit demands, the ANF problem does not admit a constant factor
        approximation  if only constant congestion is allowed \cite{ACGKTZ}.
        Thus, obtaining a good throughput approximation even in restricted
        settings requires $\omega(1)$ congestion violation.

	The ANF problem gets considerably more challenging in directed
        graphs. 
        Chuzhoy et
        al.\ ~\cite{chuzhoy2007hardness} show that,
        even if restricted to unit
        demands, the throughput is hard to approximate to within 
        polynomial factors in directed graphs when constant congestion is allowed.
        In~\cite{ChekuriE15}, Chekuri and Ene consider a variation of
        the ANF problem --- the \emph{Symmetric All or Nothing Flow
        (SymANF)} problem --- in {\em directed} graphs with {\em
          symmetric unit demand pairs and unit edge capacities}, also
        aiming at constant $\beta$ and polylogarithmic $1/\alpha$. In
        SymANF, the input pairs are unordered and a pair $s_i, t_i$ is
        routed if and only if both the ordered pairs $(s_i, t_i)$ and
        $(t_i, s_i)$ are routed; the goal is to find a maximum subset
        of the given demand pairs that can be routed.  The authors
        provide a poly-logarithmic throughput approximation with constant
        congestion for SymANF, by extending the well-linked
        decomposition framework of \cite{ChekuriKS05} to the directed
        graph setting with symmetric demand pairs. However, their
        approach, like the one for undirected graphs is limited to the
        setting where $\max_i d_i \le \min_e c_e$.  As explained
        above, our work considers a more general network setting where
        demand pairs need not be symmetric and demands values can
        exceed the capacities. Further, our goal is to obtain an arbitrarily good 
        approximation of the throughput while relaxing the capacity violation
        which is different regime.

	The \emph{Maximum Edge-Disjoint Paths
        (MEDP)}~\cite{erlebach2001maximum} problem considers a set of
        pairs of nodes to be routable if they can be connected using
        edge-disjoint paths and aims at finding the largest number of
        routable pairs. The \emph{Unsplittable Flow Problem} (UFP) is
        a generalization of MEDP to non-uniform demands while requiring
        that all flow for a pair is routed along a single path.
        MEDP and UFP are classical routing problems and have been
        extensively studied in VLSI routing where the constraint of
        using a single path for connecting pairs is particularly
        important. MEDP and UFP  tend to be harder to approximate
        than ANF.
        For instance, even for unit demands
        and undirected graphs MEDP is hard to approximate to almost a 
        polynomial factor %\redcomment{AR: I am not sure I understand what "hard to approximate to almost any polynomial factor" means... Also, this seems to imply that the undirected case is harder to approximate than the directed case, which does not make sense...}\chandra{The hardness of MEDP in undirected graphs is weaker than in undirected graphs. "Almost a polynomial factor" means it is not quite $n^\delta$ yet. The reference is correct, don't worry abou this.}
        \cite{CKN18}, and in directed graphs the
        problem is hard to approximate to within an $\Omega(m^{1/2-\epsilon})$-factor
        ~\cite{GURUSWAMI2003473}. MEDP and UFP have  been mostly considered
        under the no-bottleneck assumption, that is, when $\max_i d_i \le \min_e c_e$.
        Without this assumption UFP becomes hard to approximate to within an
        $m^{1/2-\epsilon}$ factor even for very restricted settings \cite{ShepherdV}.

	Finally, our work leverages randomized rounding techniques presented by Rost et al.~\cite{ifip18round,ccr19tw} in the different context of virtual network embedding problems
	(i.e., in their context, flow endpoints are subject to optimization).

	\subsection{Organization}%
	
	The remainder of the paper is organized as follows.
	We introduce our packing framework in Section~\ref{sec:anf-packing}
	and the compact edge-flow formulation in Section~\ref{sec:compact}.
    The multiplicative-weight-update (MWU) algorithm is described in Section~\ref{sec:mwu}, 
    our randomized rounding algorithm in Section~\ref{sec:randomized-algorithm}, and our derandomized algorithm in Section~\ref{sec:derandomization}. We discuss more general applications of our packing framework in Section~\ref{sec:extensions}.
    We report on simulation results in Section~\ref{sec:sim},
    and conclude in Section~\ref{sec:concl}.
	
\newcommand{\lpopt}{\text{OPT}_{\text{LP}}}
\newcommand{\lpval}{W_{\text{LP}}}
\newcommand{\ipopt}{\text{OPT}_{\text{IP}}}

\section{A Packing Framework for ANF}
\newcommand{\wmax}{w_{\max}}
\newcommand{\cE}{\mathcal{E}}
\label{sec:anf-packing}

We develop two non-trivial mixed integer programming (MIP)
formulations for the ANF problem, presented in this section and in Section~\ref{sec:compact}. In our approach, we compute their
linear programming~(LP) relaxation solutions in polynomial time and
then convert these solutions into integer solutions via appropriate
randomized rounding. In this section we present the first such MIP
formulation, that takes a ``packing'' view of the ANF problem and
allows for a more flexible approach, as we discuss below and in
Section~\ref{sec:extensions}. In this formulation, we will be packing
an entire flow assignment for each commodity $i$, selected from the
set of all possible valid flows between $s_i$ and $t_i$. Since the
number of possible flows will be exponential, this formulation has
exponential size, but we show that its LP relaxation can still be
solved in polynomial time via a separation oracle. This is akin to use
the path formulation for flows rather than the edge-based flow
formulation. This perspective allows one to easily see why the
randomized rounding framework for rounding paths easily generalizes to
rounding ``flows.''

Recall that the input is a directed graph $G=(V,E)$, with $n=|V|$ and
$m=|E|$, and with edge capacities $c: E \rightarrow \mathbb{Z}_+$ and
$k$ demand pairs
$(s_1,t_1),\ldots, (s_k,t_k)$. Each demand pair $i$ has an associated
non-negative weight $w_i$ and a non-negative integer demand $d_i$.  We
say that $f: E\rightarrow \mathbb{R}_+$ is a \emph{valid} flow for
pair $i$ if $f$ routes $d_i$ units from $s_i$ to $t_i$ in $G$ and
respects the edge capacities. Note that if pair $i$ cannot be routed
in isolation in $G$ then we may as well discard it (since there are no valid flows for $i$).  Let
$\mathcal{F}_i$ denote the set of all valid flows for pair $i$. Each
$\mathcal{F}_i$ is not necessarily a finite set.  However,
we can restrict attention to a finite set of flows by considering
the polyhedron of all feasible $s_i$-$t_i$ flows in  $G$ and
considering only the finitely many vertices of that polyhedron; any
valid flow can be expressed as a convex combination of the flows
defined by the polyhedron's vertices.

We now describe a mixed integer programming formulation that captures
the ANF problem. This formulation is very large: In general it can be
exponential in $n,m$ and $k$. For each $i$, we have a binary variable
$x_i$ to indicate whether commodity $i$ is routed or not.  For each
$i$ and each valid flow $f \in \mathcal{F}_i$, we have a variable
$y(f)$ to indicate the fraction of $x_i$ that is routed using the flow
$f$.  For a flow $f$ we let $f(e)$ denote the amount of flow on $e$
used by $f$; note that $f(e)$ is fixed, for each $f$ and $e$, and
hence is not a variable.
\begin{frame}{}
	\small
    \begin{figure}[ht]
        \begin{minipage}[b]{0.45\linewidth}
            \centering
                \begin{alignat*}{7} \max \sum_{i = 1}^{k} w_i x_i  \\
            	\sum_{f \in \mathcal{F}_i} y(f) = x_i \qquad && 1 \leq i \leq k \\
        	    \sum_{i = 1}^{k} \sum_{f \in \mathcal{F}_i} f(e)y(f) \leq c(e) && e \in E \\
				x_i \in \{0,1\} && 1 \leq i \leq k \\
             	y(f) \geq 0	\qquad && f \in \mathcal{F}_i, 1 \leq i \leq k
            \end{alignat*}
            $(a)$
        \end{minipage}
        \hspace{0.5cm}
        \begin{minipage}[b]{0.45\linewidth}
            \centering
                            \begin{alignat*}{7} \max \sum_{i = 1}^{k} w_i  \sum_{f \in \mathcal{F}_i} y(f) \\
            	\sum_{f \in \mathcal{F}_i} y(f) \leq 1 \qquad && 1 \leq i \leq k \\
            	\sum_{i = 1}^{k} \sum_{f \in \mathcal{F}_i} f(e)y(f) \leq c(e) && e \in E \\
            	~\\
             	y(f) \geq 0	\qquad && f \in \mathcal{F}_i, 1 \leq i \leq k
            \end{alignat*}
            $(b)$
        \end{minipage}
        \caption{$(a)$ Mixed integer programming formulation for ANF based on ``flow'' variables; $(b)$ its LP relaxation.}
            \label{fig:mip-flow}
            \label{fig:lp-flow}
    \end{figure}
\end{frame}
The following lemma is easy to see.
\begin{lemma}
  The formulation shown in Fig~\ref{fig:mip-flow}(a) is an exact
  formulation for the ANF problem.
\end{lemma}

We will now focus on solving and rounding the LP relaxation of the
preceding formulation; we simplify it by eliminating the variables
$x_i$. See Fig~\ref{fig:lp-flow}(b).

\subsection{Solving the Packing LP Relaxation}
\label{sec:solving-packingLP}
It is not at first obvious that the LP relaxation of the ANF MIP can be
solved in polynomial time. There are two ways to see why this is
indeed possible.  One is to show via the Ellipsoid method that the
dual has an {\em efficient separation oracle for the dual LP} and the
other is to describe an {\em equivalent compact (polynomial-size)
  formulation} to the ANF LP.  In this section, we will present the
former approach, which gives us a more flexible formulation that leads
to interesting extensions and that also leads to simpler proofs. In
Section~\ref{sec:compact}, we will present our strengthened compact edge-flow formulation, of
size polynomial in $n$ and $k$, and show that its relaxation is
equivalent to the relaxation of the formulation in
Figure~\ref{fig:lp-flow}(b). The benefits of the compact formulation are
that it directly leads to simple randomized and derandomized algorithms,
that can be efficiently implemented, as we show in
Section~\ref{sec:sim}.

In Figure~\ref{fig:dual-lp}, we present the dual LP to the formulation
in Figure~\ref{fig:lp-flow}(b).  There are two types of variables: First,
for each of the capacity constraints, we associate a dual variable
$\ell_e$ and for each constraint limiting the total flow to $1$ we
associate a dual variable $z_i$.
(Recall that the value  $f(e)$ is a constant
and not  a variable.)

\begin{figure}[htb]
  \begin{center}
    \begin{align*} \min \sum_{e \in E} c(e) \ell_e + \sum_{i=1}^k z_i \\
      z_i + \sum_{e \in E} f(e) \ell_e & \ge w_i \quad \quad 1\le i
                                         \le k, f \in  \mathcal{F}_i\\
      \ell_e  &  \ge 0 \quad \quad e \in E \\
      z_i  &  \ge 0 \quad \quad 1 \le i \le  k
    \end{align*}
\end{center}
  \caption{Dual of the LP relaxation for ANF.}
  \label{fig:dual-lp}
\end{figure}

The following lemma shows that one can use a polynomial-time separation oracle for solving the dual LP.

\begin{lemma}
  There is a polynomial-time separation oracle for the dual LP.
\end{lemma}
\begin{proof}
  The dual LP is easily seen to reduce to $s$-$t$ minimum-cost flow. Given
  non-negative values for the variables $\ell_e, e \in E$ and
  $z_i, 1 \le i \le k$ we compute the minimum-cost flow for each pair
  $(s_i,t_i)$ of $d_i$ units with edge costs given by
  $\ell_e, e \in E$. Let this cost be $q_i$. The values are feasible
  for the dual iff $z_i + q_i \ge w_i$ for $1 \le i \le k$. If there
  is an $i$ for which $z_i + q_i < w_i$ the corresponding minimum cost
  flow $f$ for pair $i$ defines the violated constraint.  THe minimum-cost
  flow problem is poly-time solvable and hence there is a poly-time separation
  oracle for the dual LP.
\end{proof}

Standard techniques allow one to solve the primal LP from an optimum
solution to the dual LP. However, since the Ellipsoid method is
impractical, in Sections~\ref{sec:compact} and \ref{sec:mwu}, we present two efficient ways of solving the ANF packing LP in practice,
which we will use in our implementations.

\subsection{Rounding the Packing LP Relaxation}
\label{sec:packing-rouding}
In this section, we show how to round a (fractional) solution to the
primal ANF MIP formulation. We will need the following standard
Chernoff bound (see \cite{MotwaniR}):

\begin{theorem} \label{chernoff} Let $X_1, \dots, X_n$ be $n$
independent random variables (not necessarily distributed
identically), with each variable $X_i$ taking a value of $0$ or $v_i$
for some value $0 < v_i \leq 1$.  Let $X = \sum_{i = 1}^n X_i$ be their sum.
Then the following hold:
\begin{itemize}
  \item For $\mu \ge E[X]$ and $\delta > 0$,
    $\Pr [X \geq (1 + \delta)\mu] < \bigg( \frac{e^{\delta}}{(1+ \delta)^{(1+ \delta)}} \bigg)^\mu.$
  \item For $0 \le \mu \le E[X]$ and $\delta \in (0,1)$,
        $\Pr [X \leq (1 - \delta)\mu] < e^{-\delta^2 \mu/2}.$
\end{itemize}
\end{theorem}

Randomly rounding a feasible solution to the LP relaxation is
straightforward, and is very similar to the standard rounding via the
path formulation for the Maximum Edge Disjoint Problem (MEDP) pioneered in
the work of Raghavan and Thompson~\cite{RaghavanT}. Once the LP
relaxation is solved, we consider the support of the solution. For
each pair $i$, the LP relaxation identifies some $h_i$ flows
$f_{i_1}, f_{i_2}, \ldots, f_{i_{h_i}} \in \mathcal{F}_i$ along with
non-negative values $y(f_{i_1}),\ldots,y(f_{i_{h_i}})$ such that their
sum is at most $1$. The randomized algorithm simply picks for each $i$
independently, at most one of the flows in its support where the
probability of picking $f_{i_j}$ is exactly $y(f_{i_j})$.  Note that
the probability that one chooses to route pair $i$ is exactly
$\sum_{f \in \mathcal{F}_i} y(f)\leq 1$.

We will analyze the algorithm with respect to the weight of the
LP solution $\sum_{i=1}^k w_i \sum_{j=1}^{h_j}y(f_{i,j})$. We
refer to this quantity as $\lpval$. We refer to the value of
an optimum LP solution as $\lpopt$ and the value of an optimum
integer solution as $\ipopt$. We observe that $\lpopt \ge \ipopt$
and $\lpopt \ge \lpval$. Note that when solving the formulation in Figure~\ref{fig:lp-flow}(b) or the compact formulation presented in Section~\ref{sec:compact}, the LP solution obtained will be optimal and hence $\lpval = \lpopt$; however, the solution obtained via the multiplicative-weight update algorithm of Section~\ref{sec:mwu} may only approximate $\lpopt$ and hence one could indeed have $\lpopt> \lpval$. We will also assume that $\lpopt\geq \wmax$, since we can discard from consideration any commodity $i$ that cannot be routed alone in the network, as it will never be part of a feasible solution of the MIP formulation, and hence $\wmax\leq \ipopt\leq \lpopt$.

\begin{lemma}
  \label{lem:random-weight}
  Let $Z$ be the (random) weight of the pairs chosen to be routed by
  the algorithm.  Then $E[Z] = \lpval$ and $\Pr[Z < (1-\delta) \lpval]
  < e^{-\frac{\delta^2}{2}\frac{\lpval}{\wmax}}$.  In particular,
  $\Pr[Z < (1-\delta) \lpval] < e^{-\delta^2/2}$.
\end{lemma}
\begin{proof}
  Let $Y_i$ be the indicator for pair $i$ being chosen to be routed.
  We have $Z = \sum_{i=1}^k w_i Y_i$.  The rounding algorithm implies
  that $\Pr[Y_i=1] = \sum_{f \in \mathcal{F}_i} y(f)$.  Hence, by
  linearity of expectation, $E[Z] = \sum_i w_i E[Y_i] = \sum_i w_i
  \sum_{f \in \mathcal{F}_i} y(f) = \lpval$.  Let $Z_i = \frac{w_i}{\wmax} Y_i$;
  note that $Z_i \le 1$ and $\sum_i Z_i = \frac{1}{\wmax} Z$.
  Let $Z' = \sum_i Z_i$. Since $Z'$ is a sum of independent random variables,
  each of which in $[0,1]$, we can apply the lower-tail Chernoff bound
  for $Z'$, and obtain a lower-tail bound for $Z$.
  $$\Pr[Z < (1-\delta) \lpval] = \Pr[Z' < (1-\delta) \lpval/\wmax]
  = \Pr[Z' < (1-\delta) E[Z']] \le e^{-\frac{\delta^2}{2} (\lpval/\wmax)}.$$
\end{proof}
\begin{lemma}
  \label{lem:overflow-prob}
  For $m \ge 9$ and $b > 1$ the probability that the total flow on an
  edge $e$ is more than $ (3b \ln m/ \ln \ln m) c(e)$ is at most
  $e^{-1.5 b \ln m - \frac{3b \ln b \ln m}{\ln \ln m} - 1}$. Via the union bound, the probability that the
  total flow on \emph{any} edge $e$ is more than $(3b \ln m/\ln \ln m)
  c(e)$ is at most $e^{-(1.5 b -1)\ln m - \frac{3b \ln b \ln m}{\ln \ln m} - 1}$.
\end{lemma}
\begin{proof}
  Let $X_{e}$ be the random variable indicating the total flow on edge
  $e$. Let $X_{e,i}$ be the flow on $e$ from the flow chosen for pair
  $i$. We have $X_e = \sum_{i=1}^k X_{e,i}$ and moreover the variables
  $X_{e,i}$, $1 \le i \le k$ are independent by the algorithm.  Note
  that $0 \le X_{e,i} \le c(e)$ since each flow in $\mathcal{F}_i$ is
  a valid flow by definition. Further
  $$E[X_e] = \sum_{i} E[X_{e,i}] = \sum_{i=1}^k \sum_{f \in
    \mathcal{F}_i} f(e) y(f) \le c(e).$$
  We now apply the Chernoff bound to see that $\Pr[X_e > (3b \ln m/\ln \ln m)c(e)] \le \frac{e^{\delta}}{(1+ \delta)^{(1+ \delta)}}$ where
  $(1+\delta) = 3b \ln m/\ln \ln m$; we note that the standard bound
  has all variables bounded in $[0,1]$ while all our variables are
  in $[0,c(e)]$ but we can simply scale all variables by $c(e)$.
  We have $\frac{e^{\delta}}{(1+ \delta)^{(1+ \delta)}} = e^{\delta - (1+\delta)\ln (1+\delta)}$. We consider the expression $\delta - (1+\delta)\ln (1+\delta)$, where
  $(1+\delta) = 3b \ln m/\ln \ln m$.

  \begin{eqnarray*}
    \delta - (1+\delta)\ln (1+\delta) & = & (1+\delta)(1 - \ln (1+\delta)) - 1\\
    &=& (3b \ln m/\ln \ln m) (1 - \ln (3b \ln m/\ln \ln m))) - 1\\
    & = & (3b\ln m/\ln \ln m) (1 - \ln 3 - \ln b - \ln \ln m + \ln \ln \ln m) - 1 \\
    & \le & (3b\ln m/\ln \ln m) (- \ln b - \frac12 \ln \ln m) - 1 \\
    & \le & -1.5 b \ln m  - 3b \ln b \ln m/\ln \ln m - 1.
  \end{eqnarray*}
  In the above we used the fact that $\ln \ln m - \ln \ln \ln m \ge
  \frac12\ln\ln m$ for $m \ge 9$.

  The second part follows easily via the union bound over all the $m$ edges.
\end{proof}

We can now put together the preceding lemmas to derive our bicriteria
approximation. We will henceforth assume that $m \ge 9$.  Let $S$ be
the random set of pairs routed by the algorithm.  Let $\cE_1$ be the
event that $w(S) < (1-\delta)\lpval$. From
Lemma~\ref{lem:random-weight}, $\Pr[\cE_1] \le e^{-\delta^2/2}$.  Let
$\cE_2$ be the event that there is some edge $e$ such that the flow on
$e$ is more than $(3b \ln m/\ln \ln m)c(e)$.  From
Lemma~\ref{lem:overflow-prob} $\Pr[\cE_2] \le e^{-1.5 b \ln m - 3b \ln
  b \ln m/\ln \ln m - 1}$. For $b=1$ and $m \ge 9$ we see that
$\Pr[\cE_2] \le e^{-12}$. Chossing $\delta = 1/2$, $\Pr[\cE_1] \le
0.8825$. Thus $\Pr[\cE_1 \text{ or } \cE_2] \le 0.9$.  This implies
that with probability $\ge 0.1$, the set $S$ of routed pairs satisfies
the property that $w(S) \ge 0.5 \lpval$ \emph{and} the congestion of
every edge is at most $3 \ln m /\ln \ln m$. In other words, if
$\lpval = \lpopt \le \ipopt$, we obtain
a $(1/2,3 \ln m/\ln \ln m)$-bicriteria approximation with probability
at least $0.1$.  One can boost the success probability by repetition.
If the rounding is repeated $10 c \ln m$ times, then with probability
at least $1- (1 - 0.1)^{10 c \ln m} \ge 1- m^{-c}$ (in other words with
high probability), one of the rounded solutions is a $(1/2,3 \ln m/\ln
\ln m)$-bicriteria approximation

We now refine the preceding argument to show that the quality of the
rounded solution can get arbitrarily close to $\lpval$ but with lower
probability, and examine the trade-off required in the congestion and
number of repetitions required.  Suppose we want $w(S) \ge (1-\eps)
\lpval$ for some small $0 < \eps < 1/2$. Let $\cE_1$ be the event
that this does not happen. From Lemma~\ref{lem:random-weight}, we have that
$\Pr[\cE_1] \le e^{-\eps^2/2}$. 
Let $\cE_2$ be the
event that some edge congestion exceeds $3b \ln m/\ln \ln m$.
Lemma~\ref{lem:overflow-prob} allows us to upper bound this
probability.
Suppose we choose $b$ such that $\Pr[
{\cE_2}] \le \eps^2/6$.
Then 
$$\Pr[\bar{\cE_1} \cap \bar{\cE_2}] = 1 - \Pr[\cE_1 \cup \cE_2] \ge (1- \Pr[\cE1] - \Pr[\cE_2]) \ge 1-e^{-\eps^2/2} - \eps^2/6 \ge \eps^2/6.$$
This would yield a $(1-\eps,3b \ln m/\ln \ln m)$ bicriteria approximation
with probability at least $\eps^2/6$ and one can boost this via
repeating $O(\frac{1}{\eps^2} \ln m)$ times to get the approximation
with high probability. Thus it remains to estimate $b$ such
that $\Pr[\cE_2] \le \eps^2/6$. From Lemma~\ref{lem:overflow-prob},
it suffices to choose $b$ such that
$$(1.5b-1)\ln m + 3b \ln b \ln m/\ln m \ln m +1 \ge \ln (6/\eps^2).$$
In particular it suffices to have
$b \ge c \frac{\ln(1/\eps)}{\ln m}$ for some fixed constant $c$.
Thus for all $\eps \ge 1/m$ a fixed constant $b$ (e.g., $b=1.85$) suffices!
Note however that the number of repetitions grows as $\Omega(1/\eps^2)$
to guarantee a good solution with high probability.

\begin{theorem}
	\label{thm:main-theorem-without-alteration}
  For $m \ge 9$ and any $1/m\leq \eps <1$,
  %fixed $0<\eps <1$ 
  there is a polynomial-time randomized
  algorithm that yields a $(1-\eps, O(\ln m/\ln \ln m + \ln (1/\eps)
  /\ln m))$-approximation with high probability. Moreover, by setting $\eps = 1/m$, we guarantee a
  $O(1-1/m, O(\ln m/\ln \ln m))$-approximation with high probability.
\end{theorem}
 
Noting that it is trivial to get a $(1,k)$-approximation by simply routing all the commodities at full demand, we get the following corollary, stating our full approximation guarantees:

\begin{corollary}
\label{cor:bestapprox}
For $m \ge 9$ and any $1/m\leq \eps <1$,
%any $\eps \geq 1/m$ (or any fixed $0<\epsilon<1$) 
there is a polynomial-time randomized
  algorithm that yields a $(1-\eps, \min\{k,O(\ln m/\ln \ln m)\}))$-approximation with high probability.
\end{corollary}

%\redcomment{AR: I think we need to be consistent on how exactly we quantify $\eps$ throughout the paper (including in abstract and intro). I propose that we use $1/m\leq \eps < 1$ everywhere (please let me know if there should be exceptions to the rule; I could not see any). I changed it according what I am proposing in the theorem and corollary above (you can see what was there before as commented text); this is also how the abstract and intro quantify $\eps$. Let me know if you prefer something different.} \chandra{I am ok with this. $\eps < 1/m$ is too small anyway. Typically we think of $\eps$ as a fixed constant as $m$ increases.}
 %Due to space limitations,  
 We describe a different rounding
approach in Appendix~\ref{sec:alteration}, using an {\em alteration approach}, that may also be of interest in certain settings and gives a
better tradeoff in terms of repetitions.

\section{Compact Edge-Flow Formulation}
\label{sec:compact}
As we saw, one can solve the ANF packing LP via the Ellipsoid method. While this leads to a polynomial-time algorithm for solving the LP,  implementing such algorithm would not be trivial nor be very efficient in practice. In this section, we present an alternative polynomial-size compact edge-flow formulation for the ANF problem, which can be solved more efficiently in practice than the packing LP. In Section~\ref{sec:mwu}, we present another approach for solving the packing LP more efficiently, albeit only approximately. Both approaches 
were evaluated in simulations in Section~\ref{sec:sim}.

We present our general compact edge-flow based MIP formulation for the ANF problem\footnote{The compact MIP formulation presented here generalizes the one in our conference version~\cite{Liu19} to accommodate arbitrary demands and commodity weights.} in Figure~\ref{fig:compact}. 
We use an indicator variable $f_{i} \in \{0,1\}$ to indicate whether a commodity $i$ is successfully routed through $G$.
Next, we denote $f_{i,e} \in [0,1]$ as the fraction of flow for commodity $i$ allocated to a particular edge $e \in E$.
The total flow assigned to a fixed edge $e$ is given by $\sum_i d_i\cdot f_{i,e}$ and the total weighted throughput is given by $\sum_{i}w_{i}f_i$. 
Constraints~(\ref{const:value}-\ref{const:edgecap}) define the value of the total flow for each commodity $i$, 
%(\ref{const:flowcons}) 
enforce flow conservation for each $i$, and 
%(\ref{const:edgecap}) 
stipulate that no edge capacity is violated by the flow assignments. 
Constraint~(\ref{const:strenghten}) ensures that for a fixed commodity $i$, the ratio of flow assigned to an edge $e$ to the total flow of that commodity does not exceed the capacity of $e$: These constraints are actually redundant for the MIP formulation, but will {\em strengthen the LP relaxation} of Figure~\ref{fig:compact}, obtained by allowing each $f_i$ to assume any real value in $[0,1]$. In fact, (\ref{const:strenghten}) is crucial to establish the perhaps surprising "equivalence" between the LP relaxation of the ANF packing formulation
(Figure~\ref{fig:lp-flow}(b)) and the LP relaxation of the compact edge-flow MIP.
	
This formulation has size polynomial in $n$ and $k$ and hence can be solved in polynomial time (e.g., using the Ellipsoid method). Moreover, given the compact nature of the LP, one can use a standard LP solver in practice.

\begin{figure}[htb]

\begin{center}
		\begin{align}
	    \max \sum_{i=1}^{k}w_{i}f_{i} &\mbox{\ \ \ }&\ \  \\ 
		\sum_{(s_i,v) \in E}f_{i, (s_i,v)} &= f_i  &\forall i \in [k] \label{const:value}\\
		\sum_{(u, v) \in E}f_{i,(u, v)} &= \sum_{(v, u) \in E}f_{i, (v, u)} & \forall i \in [k] , \forall v \in V - \{s_i,t_i\} \label{const:flowcons}\\
		\sum_{i=1}^{k}f_{i,(u, v)}\cdot d_{i} &\leq c_{(u, v)} & \forall (u, v) \in E \label{const:edgecap}\\
		f_{i,(u, v)}\cdot d_{i} &\leq f_i \cdot c_{(u, v)} & \forall i \in [k], \hfill  \forall (u, v) \in E \label{const:strenghten}\\
		f_{i,(u, v)} &\geq 0 & \forall i \in [k], \forall (u, v) \in E \label{const:nonneg} \\
	    f_i &\in \{0, 1\} &\forall i \in [k] \label{const:int}
		\end{align}
		\end{center}
		\caption{Compact Edge-Flow ANF Formulation}
		\vspace{-12pt}
		\label{fig:compact}

\end{figure}

\label{sec:compact-equivalence} %
\paragraph
{\bf Equivalence with Packing Formulation.}
Here we prove that the packing formulation in Figure~\ref{fig:lp-flow}(b) is ``equivalent'' to the
compact formulation given in Figure~\ref{fig:compact}. When we say equivalent we mean the following:
Given a feasible solution to one LP we can obtain a feasible solution to the other LP of the same value.
We prove both directions below.

First, consider a feasible solution to the compact formulation. For commodity $i$, let
$f_i \in [0,1]$ be the total fraction of $d_i$ that is routed from $s_i$ to $t_i$, and let
$f_{i,e} \in [0,1]$ be fraction of $f_i$ that is assigned to edge $e \in E$, satisfying flow conservation and capacity constraints.  We first construct a
flow $g_i: E \rightarrow \mathbb{R}_+$ of $d_i$ units from $s_i$ to
$t_i$: We set $g_i(e) = d_i f_{i,e}/f_i$, for all $e\in E$. It is easy to
verify that $g_i$ is a flow of $d_i$ units from $s_i$ to
$t_i$. Moreover by the strengthening constraint
(\ref{const:strenghten}) in Figure~\ref{fig:compact}, we see that
$g_i(e) \le c(e)$ for all $e$ and hence $g_i$ is a feasible flow in
the capacities. Putting together these facts, $g_i \in \mathcal{F}_i$.
We obtain a feasible solution to the packing formulation as follows.
For each $i$ we set $x_i = f_i$ and we set $y(f) = x_i$ for $f = g_i$
and $y(f) = 0$ for every other $f \in \mathcal{F}_i$. In other words
we are using only one flow for each commodity $i$. The only
non-trivial fact to check is that this solution is feasible.
For this we need to verify that $\sum_i y(g_i) g_i(e) \le c(e)$ but this
easily follows from our definition of $g_i$'s and Constraint
(\ref{const:edgecap}) in Figure~\ref{fig:compact}. Since $x_i=f_i$
for all $i$, we see that the two solutions have the same  value.

Second, consider a feasible solution $y$ to the packing formulation
in Figure~\ref{fig:lp-flow}(b). Let $x_i$ be the amount routed for
commodity $i$ and for each flow $f \in \mathcal{F}_i$, $y(f)$ is the
amount routed on $f$ with $\sum_{f \in \mathcal{F}_i} y(f) = x_i$.  We
construct a feasible solution to the compact LP as follows.  For each
commodity $i$ we set $f_i = x_i$. For each $e \in E$ and each $i \in
[k]$, we set $f_{i,e} = \frac{1}{d_i}\sum_{f \in \mathcal{F}_i}f(e)y(f)$.
Note that $f_{i,e}$ is simply scaling by $d_i$ the total flow on $e$ from
all $f \in \mathcal{F}_i$.  Since each $f \in \mathcal{F}_i$ is a flow
of $d_i$ units from $s_i$ to $t_i$ and $\sum_{f \in \mathcal{F}_i}
y(f) = x_i$ we see that $f_{i,e}$, $e \in E$, corresponds to sending a
total of $x_i$ units of flow from $s_i$ to $t_i$.  We focus on
Constraints (\ref{const:edgecap}) and (\ref{const:strenghten}) in
Figure~\ref{fig:compact}. We observe that $\sum_i d_i f_{i,e} = \sum_i
d_i 
%AR: I commented this out since it should not be therer, hope I did not miss anything
%\sum_{f \in \mathcal{F}_i} 
\frac{1}{d_i}\sum_{f \in
  \mathcal{F}_i}f(e) y(f) = \sum_i \sum_{f \in \mathcal{F}_i}y(f)f(e)$
and the last term is at most $c(e)$ from the feasiblity of given
solution for the packing formulation.  This proves that Constraint
(\ref{const:edgecap}) in Figure~\ref{fig:compact} is satisfied for
the solution we constructed. We observe that for each $f \in
\mathcal{F}_i$ and each $e \in E$ we have $f(e) \le c(e)$ since $f$ is
a feasible flow in the capacities. Thus $f(e)/d_i \le c(e)/d_i$ and
since $y(f) \ge 0$ for each $f \in \mathcal{F}_i$ we have $\sum_{f \in
  \mathcal{F}_i} y(f) f(e)/d_i \le c(e)/d_i \sum_{f \in \mathcal{F}_i}
y(f)$ which implies that $f_{i,e} d_i \le f_i c(e)$. Thus
the solution also satisfies (\ref{const:strenghten}) in
Figure~\ref{fig:compact}. 
This finishes the proof of the equivalence.

Hence, the results in Section~\ref{sec:anf-packing} that lead to Theorem~\ref{thm:main-theorem-without-alteration}
and Corollary~\ref{cor:bestapprox} also apply to a randomized rounding approach based on the compact formulation, as we explain in Section~\ref{sec:randomized-algorithm}.

\section{MWU Algorithm}
\label{sec:mwu}
While the compact edge-flow formulation can always be solved in polynomial time, one may run into space issues when attempting to solve it in practice:
 The disadvantage of using a
standard LP solver to solve the compact edge-flow LP relaxation is that the number
of variables is $km$ which is quadratic in the input size, and the
number of constraints is $m$.  Standard LP solvers often require space
proportional to $km^2$ which can be prohibitive even for moderate
instances (since it is almost cubic in input size). One advantage of
the packing LP formulation, although it is equivalent, to the compact
formulation is that one can use well-known multiplicative weight
update (MWU) based Lagrangean relaxation approaches to obtain a
$(1-\gamma)$-approximation, for any $0<\gamma<1$. Although the convergence
time can be slow depending on the accuracy required, the space
requirement is $O(k+m)$ which is linear in the input size. In
addition, there are several optimization heuristics based on the MWU
algorithm that can result in very efficient implementations in
practice.  Since the MWU framework is standard we only describe and
explain the algorithm here and state the known
guarantees on the number of iterations and time complexity, referring the reader to standard treatments
in the literature~\cite{AHK-survey} for a formal analysis on the correctness
guarantees.

\begin{figure}[htb]
	\begin{subfigure}[h]{0.5\linewidth}
		\begin{center}
			\begin{align*} (a) \max \sum_{i=1}^k w_i \sum_{f \in \mathcal{F}_i} y(f) & \\
				\sum_{i=1}^k \sum_{f \in \mathcal{F}_i} f(e) y(f)  & \le c(e)
				\quad \quad
				e \in E\\
				y(f) &  \ge 0 \quad \quad f \in \mathcal{F}_i, 1 \le i  \le k.
			\end{align*}
		\end{center}
	\end{subfigure}
	\begin{subfigure}[h]{0.4\linewidth}
		\begin{center}
			\begin{align*} (b) \min \sum_{e \in E} c(e) \ell_e \\
				\sum_{e \in E} f(e) \ell_e & \ge w_i \quad \quad 1\le i
				\le k, f \in  \mathcal{F}_i\\
				\ell(e)  &  \ge 0 \quad \quad e \in E \\
			\end{align*}
		\end{center}
	\end{subfigure}
	\caption{$(a)$ LP Relaxation with no constraint on total amount routed per commodity; $(b)$ its dual.}
	\label{fig:lp-flow-simplified}
\end{figure}

\smallskip
\noindent {\bf Algorithm description:} MWU based algorithms are
iterative and provide a way to obtain arbitrarily good relative
approximation algorithms for a large class of linear programs such as
packing, covering and mixed packing and covering LPs. In particular we
can apply it to the packing LP in Fig~\ref{fig:lp-flow}(b). The LP has two
types of packing constraints, one involving the capacities, and the
other involving the total amount of flow routed for each commodity. It
is useful to simplify the LP further in order to apply a clean packing
framework. For this purpose we alter the given graph $G=(V,E)$ as
follows. For each given demand pair $(s_i,t_i)$ we add a dummy source
$s'_i$ and connect it to $s_i$ with an edge $(s'_i,s_i)$ of capacity
equal to $d_i$. We replace the pair $(s_i,t_i)$ with the pair
$(s'_i,t_i)$, which ensures that the total amount of flow for the pair
is at most $d_i$, and further allows us to eliminate the first set of
constraints in Fig.~\ref{fig:lp-flow}(b). In the modified
instance we hence only have edge capacity constraints and the problem
becomes a pure maximum throughput problem that allows for a commodity
to be routed more than one total unit. The dual LP also simplifies in
a corresponding fashion. These are shown in
Fig~\ref{fig:lp-flow-simplified}. %

\begin {figure}[htb]
\begin{algorithm}[H]
	\caption{MWU for Multi-Commodity ANF Problem}
	\label{protocol}
	\begin{algorithmic}[1]
		\Statex{\textbf{Inputs:} Directed graph G(V,E), $c:E \to \mathbb{R}^{+}$, a set $S$ of $k$ pairs of commodities $(s_{i},t_{i})$ each with demand $d_{i}$ and  $\gamma \in \mathbb{R}^{+}$}
		\State{Change $G$ by adding dummy terminal $s'_i$ and edge$(s'_i,s_i)$ with capacity $d_i$. This ensures that we don't route more than $d_i$ units for pair $i$. We will assume this has been done and simply use $(s_i,t_i)$ instead of $(s'_i,t_i)$}
		\Statex{\textbf{Output:} Total flow $f_e$ on each $e$. $f(s'_i,s_i)/d_i$ gives the fraction of commodity $i$ that is routed}
		\State{Define a length/cost function $\ell:E \to \mathbb{R}^{+}$ and initialize $\ell_{e} \leftarrow 1, \forall e \in E$}
		\State{Define a function $f:E \to \mathbb{R}_{\geq 0} $ and initialize $f_{e} \leftarrow 0, \forall e \in E$}
		\State{Define $\eta \leftarrow \frac{\ln{|E|}}{\gamma} $}
		\Repeat
		\For{each commodity $i \in S$}
		\State{Compute min-cost flow of $d_i$ units from $s_i$ to $t_i$ with capacities $c(e)$ and cost given by $\ell$. (If no feasible flow then pair $i$ can be dropped.) Let this flow be defined by $g_i(e), e \in E$ and let cost of this flow be $\rho(i) = \sum_e \ell(e) g_i(e)$}
		\EndFor
		\State{Set $i^{*} \leftarrow \text{argmin}_{i\in S} \frac{\rho(i)}{w_i}$}
		
		\State{Compute $\delta \leftarrow \min_{e} \frac{\gamma}{\eta} \cdot \frac{g_{i^*}(e)}{c(e)}$}
		\For{ each $e$}
		
		\State{Set $f_{e} \leftarrow f_{e} + \delta g_{i^*}(e)$}
		\If{$f_{e} > c_{e}$}
		\State{Output $f$ and halt}
		
		\Else
		\State{Update $\ell_{e} \leftarrow \exp(\eta f_{e}/c_{e})$}
		\EndIf
		\EndFor
		\Until {termination}
	\end{algorithmic}
	\label{alg:MWU}
\end{algorithm}
%\vspace{-18pt}
\end{figure}

The MWU Algorithm~\ref{alg:MWU} solves the primal LP in
Fig~\ref{fig:lp-flow-simplified} in an iterative fashion as
follows. It takes as input an error parameter $\gamma \in (0,1)$ and its
goal is to output a feasible solution of value at least $(1-\gamma)$
times the optimum LP solution value.  Note that the primal LP has an
exponential number of variables but only $m$ non-trivial constraints
corresponding to the edges, so it maintains only an implicit
representation of the primal variables. The MWU algorithm can be
viewed as primal-dual algorithm as well and as such it maintains 
``weights'' (hence the name mutiplicative weights udate) for each edge
$e$ which correspond to the dual variables $\ell(e)$.  To avoid
confusion with the weights of commodities we use the term lengths. The
algorithm maintains lengths $\ell(e), e \in E$ which are initialized
to $1$. The algorithm roughly maintains the invariant that $\ell(e)$
is exponential in the current total flow $g(e)$ on edge $e$; more
formally, for a parameter $\eta = \ln m/\gamma$ the algorithm maintains
the invariant that $\ell(e) \simeq \exp(\eta f(e)/c(e))$ where $f(e)$
is the total flow on $e$.  
In each iteration the goal is to find a
good commodity/pair to route. To this end the algorithm computes for
each commodity $(s_i,t_i)$ a minimum-cost $s_i$-$t_i$ flow of $d_i$
units where the cost on $e$ is equal to $\ell(e)$. Let this cost be
$\rho(i)$. It then chooses the commodity $i^*$ that has the smallest
$\rho(i)/w_i$ ratio among all pairs, as the currently best commodity
to route. The algorithm then routes a small amount for $i^*$ along the
minimum cost flow computed in that iteration. This corresponds to the
step size $\delta$ which is chosen to be sufficiently small but not
too small to ensure the correctness of the algorithm. After routing
the flow for $i^*$ the lengths on the edges are updated to reflect the
increase in flow on the edges. The algorithm proceeds in this fashion
for several iterations until termination. One can terminate using
several different criteria while ensuring correctness. Here we
stop the algorithm when we try to route a commodity with the given
step size and realize that it violates some edge capacity.

\smallskip
\noindent {\bf Analysis of iterations, run-time and space:} The
algorithm's running time is dependent on the time to compute minimum-cost flow
and on the total number of iterations. It is known that the MWU
algorithm, as suggested above, terminates in $O(m \log m/\gamma^2)$
iterations. Each iteration requires computing $k$ minimum-cost flows.
Many algorithms are known for minimum-cost flow ranging from strongly
polynomial-time algorithms to polynomial-time scaling algorithms as
well as practically fast algorithms based on network-simplex.  Instead
of listing these we can upper bound the run-time by
$O(\text{MCF}(n,m) k m \log m/\gamma^2)$ where $\text{MCF}(n,m)$ is
min-cost flow running time on a graph with $n$ nodes and $m$ edges.  In
terms of space we observe that the algorithm only maintains the total
flow on each edge and for each commodity the total flow it has routed
as well as the lengths on the edges. This is $O(k+m)$. The algorithm
also needs space to compute minimum-cost flow and that depends on the
algorithm used for it. Most algorithms for minimum-cost flow use
space near-linear in the input graph.

The algorithm as described above is a plain "vanilla" implementation of
the general MWU algorithm. As such the running time is rather high and computing $k$ minimum-cost flows in each iteration is expensive. Several optimization can be done from both a theoretical
and a practical point of view. We do not discuss these issues in detail since this is not the main focus of this paper. We develop a simple
heuristic -- the {\em permutation routing} heuristic -- based on these ideas that has also theoretical justification, which will be discussed and used for the simulations
in Section~\ref{sec:sim}.

\section{Randomized Rounding Algorithm}
	\label{sec:randomized-algorithm}
	\allowdisplaybreaks

Algorithm~\ref{alg:randomrounding}
describes the randomized rounding algorithm that we will use in our simulations. This algorithm
performs randomized rounding on the total flow variables of the compact LP and therefore can be viewed as  a special case of the randomized rounding algorithm outlined in Section~\ref{sec:packing-rouding} (since we have proven that the set of feasible solutions to the compact LP can be viewed as a subset of the feasible solutions to the packing LP).
%Algorithm~\ref{alg:randomrounding} is also a generalization of the algorithm presented %
%in~\cite{infocom19} that accommodates arbitrary commodity demands and weights.
Algorithm~\ref{alg:randomrounding} leads to a simpler, more streamlined implementation (also because the randomized rounding approach will be based on a number of variables that is linear in the number of commodities) than if we were using the approach based on the rounding of the variables of the packing LP directly. We assume, as we did in Section~\ref{sec:anf-packing}, that we discard any commodity $i$ that cannot be routed by itself in $G$.

\begin{figure}[b!]
\removelatexerror
\begin{algorithm}[H]
\SetKwInOut{Input}{Input}\SetKwInOut{Output}{Output}
\SetKwFunction{ProcessPath}{ProcessPath}{}{}
\SetKwFunction{reverse}{reverse}{}{}
\SetKwFunction{LP}{LP}
\SetKwFunction{LP}{LP}

\Input{Directed graph $G(V, E)$ with edge capacities $c_{e}>0, \forall e \in E$; set of $k$ %
		pairs of commodities $(s_i, t_i)$, each with demand $d_{i}\geq 0$ and weight $w_{i}\geq 0$; $\epsilon \in (0,1]$}
\Output{The final values of $f_i$ and $f_{i,e}$ and $ \sum w_{i}f_i$}
			Let $\tilde{f}_i, \tilde{f}_{i,e}$, $\forall i\in [k], \forall e\in E$, be a feasible solution to compact LP. \label{algo:step_lp}\\
			For each $i \in [k]$, independently, set $f_i=1$ with probability $\tilde{f}_i$, otherwise set $f_{i}=0$.\label{algo:step_rounding}\\
		Rescale the fractional flow $\tilde{f}_{i,e}$ from the LP solution on edge $e$ for commodity $i$ by $\frac{1}{\tilde{f}_i}$: I.e., $f_{i,e} = \frac{\tilde{f}_{i,e}}{\tilde{f}_i}\cdot f_{i}$ and the flow for commodity $i$ on $e$ is given by $f_{i,e} d_i$.\label{algo:step_rescale}\\
			If $\sum_{i} w_{i}f_{i} \geq (1-\epsilon) \sum w_{i}\tilde{f}_i$
			and $\sum_i f_{i,e}d_i\leq (3b \ln m/\ln\ln m)c(e)$ for all $e\in E$, return the corresponding flow assignments given by $f_{i}$ and $f_{i,e}, \forall i\in [k]$ and $e \in E$. Otherwise, repeat steps \ref{algo:step_rounding} and \ref{algo:step_rescale},  $O((\ln m)/\epsilon^2)$ times. \label{algo:step_repeat}
		\caption{Randomized Rounding Algorithm}
		\label{alg:randomrounding}
\end{algorithm}
\end{figure}

We use randomized rounding to round the total fraction $\tilde{f}_i$ of $d_i$ that the compact LP routes for commodity $i$
to $f_i=1$, with probability  $\tilde{f_{i}}$, and to 0 otherwise. If we set $f_i$ to 1, then in order to satisfy flow conservation constraints (e.g., constraint (\ref{const:flowcons}) of Figure~\ref{fig:compact}),
we need to re-scale all the $\tilde{f}_{i,e}$ values by $1/\tilde{f}_i$, obtaining the flows $f_{i,e}$  (if $f_i=0$ then $f_{i,e}=0$, for all $e \in E$). We repeat Steps \ref{algo:step_rounding}-\ref{algo:step_rescale}
of Algorithm~\ref{alg:randomrounding} 
$\Theta((\ln m)/\epsilon^2)$ times or until we obtain the desired $((1-\epsilon),3b\ln m/\ln\ln m)$-approximation bounds, amplifying the probability of getting a desired outcome.

Given the equivalence that we showed between the packing and the compact LP, which implied among other things that the two LPS have optimal solutions of the same value and  that Algorithm~\ref{alg:randomrounding} corresponds to the packing randomized rounding approach described in Section~\ref{sec:anf-packing} when restricted to the subset of solutions to the compact LP,
we get the following corollary to
Theorem~\ref{thm:main-theorem-without-alteration}:

\begin{corollary}
	\label{cor:randomized-rounding-result}
Algorithm~\ref{alg:randomrounding}, when run on an optimum solution to the compact LP, achieves a 
{\em $((1-\eps),3b\ln m/\ln\ln m)$-approximation} for the ANF problem on {\em arbitrary networks} with high probability, for a suitable constant $b>1/m$, e.g. $b=1.85$, and any $1/m\leq \eps<1$. %
\end{corollary}

In our implementations, we will also run Algorithm~\ref{alg:randomrounding} using
the solution output by the MWU algorithm, which only guarantees a $(1-\gamma)$ approximation on the throughput for $\gamma\in (0,1)$.

In that case, we let $\tilde{f_{i}}=f(s'_i,s_i)/d_i$, where $f(s'_i,s_i)$ is as defined in Algorithm~\ref{alg:MWU}, and the values of $\tilde{f}_{i,e}$ are defined according to the flows chosen for each commodity $i$. Note that the throughput approximation guarantee for Algorithm~\ref{alg:randomrounding} in this case will be $(1-\epsilon)(1-\gamma)$.

Another advantage of 
Algorithm~\ref{alg:randomrounding} is that it leads to a surprisingly simple derandomized algorithm, as we will see  in Section~\ref{sec:derandomization}, that was also implemented for our simulations.

\section{An Efficient Deterministic Algorithm}
\label{sec:derandomization}
In this section, we give a derandomization of Algorithm \ref{alg:randomrounding}. Our derandomized algorithm is particularly attractive for its simplicity and efficiency in practice (see Section~\ref{sec:sim}), unlike most existing derandomized algorithms in the literature whose implementations in practice are cumbersome and ineffective. 
Our deterministic algorithm leverages the method of pessimistic estimators first introduced by Raghavan~\cite{Raghavan1988Derandomization} to efficiently compute conditional expectations, which will guide the construction of the $(\alpha,\beta)$-approximate solution. Given the analysis in Section~\ref{sec:packing-rouding} and Corollary~\ref{cor:randomized-rounding-result}, in the forthcoming analysis, we always assume $\alpha = 1-1/m$ and $\beta = 3b\ln m/\ln\ln m$ for $m \geq 9$ and $b = 1.85$.

\newcommand{\failurestyle}[1]{\mathsf{#1}}

\newcommand{\arbitraryFunction}{\ensuremath{\failurestyle{f}}}

\newcommand{\fail}[1][fail]{\ensuremath{\failurestyle{#1}}}
\newcommand{\est}{\ensuremath{\failurestyle{est}}}
\newcommand{\estP}{\ensuremath{\failurestyle{est^{\alpha}_{\beta}}}}
\newcommand{\estA}{\ensuremath{\failurestyle{est}_{\failurestyle{\alpha}}}}
\newcommand{\estB}[1][(u,v)]{\ensuremath{\failurestyle{est}^{#1}_{\failurestyle{\beta}}}}

\newcommand{\estAF}{\ensuremath{\failurestyle{est}_{\failurestyle{\alpha}}}}
\newcommand{\estBF}[1][(u,v)]{\ensuremath{\failurestyle{est}^{#1}_{\failurestyle{\beta}}}}

\RenewDocumentCommand{\estA}{O{\rvZ[1], \ldots, \rvZ[k]} O{(u,v)}}{\ensuremath{\failurestyle{est}_{\failurestyle{\alpha}}(#1)}}

\RenewDocumentCommand{\estB}{O{\rvZ[1], \ldots, \rvZ[k]} O{(u,v)}}{\ensuremath{\failurestyle{est}^{#2}_{\failurestyle{\beta}}(#1)}}

We first introduce the following notation.
Let $z_i = 0$ if Algorithm~\ref{alg:randomrounding} has not selected commodity $i$ to be routed, and let $z_i = 1$ if $i$ was admitted. Now, let  $\fail(z_1,\ldots,z_k) \to \{0,1\}$ denote the failure function of not constructing an $(\alpha,\beta)$-approximate solution, i.e., $\fail(z_1, \ldots, z_k) = 1$ if and only if the constructed solution either does not achieve an $\alpha$-fraction of the LP's (weighted) throughput or the capacity of some edge is exceeded by a factor larger than $\beta$. We use $Z_i$ to denote the $\{0,1\}$-indicator random variable for whether commodity $i$ is routed in one execution of Steps 3-4 of  Algorithm~\ref{alg:randomrounding},
i.e., $\Prob[Z_i = 1] = \tilde{f}_i$ and $\Prob[Z_i = 0] = 1- \tilde{f}_i$.
We have shown in Section~\ref{sec:packing-rouding} that $\Expec[\fail(Z_1,\ldots,Z_k)] < 1$ holds (cf. Theorem~\ref{thm:main-theorem-without-alteration}),  implying the existence of an $(\alpha,\beta)$-approximate solution. Given the above definitions, we employ the following notation to denote the conditional expectation of a function $\arbitraryFunction: \{0,1\}^k \to \{0,1\}$:
\begin{align*}
\Expec[\arbitraryFunction(z_1,\ldots,z_i, Z_{i+1},\ldots, Z_{k})] = \Prob[\arbitraryFunction(Z_1,\ldots,Z_{k})=1~|~Z_1=z_1,\ldots,Z_i=z_i]\,.\\[-21pt]
\end{align*}

\vspace{5pt}
As computing $\Expec[\fail(z_1,\ldots,z_i, Z_{i+1},\ldots, Z_{k})]$ is generally computationally prohibitive, we will now derive a pessimistic estimator $\est: \{0,1\}^k \to \mathbb{R}_{\geq 0}$, such that the following holds for all $i \in [k]$ and all $(z_1,\ldots,z_i) \in \{0,1\}^i$:

\vspace{3pt}
\begin{descriptionFixed}[2.25cm]
\item[Upper Bound] $\Expec[\fail(z_1,\ldots,z_i, Z_{i+1},\ldots, Z_{k})] \leq \Expec[\est(z_1,\ldots,z_i, Z_{i+1},\ldots, Z_{k})]$.~\hfill\tagIt{eq:pess:bound}
\item[Efficiency] $\Expec[\est(z_1,\ldots,z_i, Z_{i+1},\ldots, Z_{k})]$  can be computed efficiently.~\hfill\tagIt{eq:pess:efficiency}
\end{descriptionFixed}
\vspace{3pt} 
Furthermore, the estimator's value must initially be strictly less than 1 for the derandomization:

\vspace{3pt}
\begin{descriptionFixed}[2.25cm]
\item[Base Case] $\Expec[\est(Z_{1},\ldots, Z_{k})] < 1$ holds initially.~\hfill\tagIt{eq:pess:base}
\end{descriptionFixed}

\vspace{3pt}
In the following, we discuss how such a pessimistic estimator is used to derandomize the decisions made by the algorithm informally presented in Section~\ref{sec:anf-packing} before introducing the actual estimator $\estP$ in Lemma~\ref{lem:pessimistic-estimator}. Algorithm~\ref{alg:deterministic-rounding} first computes an LP solution just as Algorithm~\ref{alg:randomrounding},
but then uses the pessimistic estimator to guide its decision towards deterministically constructing an approximate solution. Specifically, each commodity is either routed or rejected such that the conditional expectation $\Expec[\estP(z_1,...,z_{i},Z_{i},...,Z_n)]$ is minimized. Given that initially $\Expec[\estP(Z_{1},\ldots, Z_{k})] < 1$, this procedure terminates with a solution $(z_1,\ldots,z_k)$ such that the failure function $\fail(z_1,\ldots,z_k)$ is strictly upper bounded by $1$. Specifically, $1 > \Expec[\estP(Z_{1},\ldots, Z_{k})] \geq \Expec[\estP(z_1,Z_{2},\ldots, Z_{k})] \geq \ldots \geq \Expec[\estP(z_1,\ldots,z_k)]$ is guaranteed and therefore, for the binary function $\fail$, $\fail(z_1,\ldots,z_k) = 0$ must hold. Furthermore, the algorithm is efficient (i.e., runs in polynomial time) as long as the pessimistic estimator function $\estP$ can be evaluated in polynomial time. W.l.o.g., we assume that only commodities that can be satisified in $G$ are given as input to Algorithm~\ref{alg:deterministic-rounding}.

\newlength{\commentWidth}
\setlength{\commentWidth}{8.5cm}
\newcommand{\rtcp}[1]{\tcp*[r]{\makebox[\commentWidth]{#1\hfill}}}
\newcommand{\ftcp}[1]{\tcp*[f]{\makebox[\commentWidth]{#1\hfill}}}

\newcommand{\failureestimate}{\ensuremath{\mathsf{failure\_estimate}}}

  \begin{figure}[t!]

  \begin{algorithm*}[H]
  \DontPrintSemicolon
  
  \SetKwInOut{Input}{Input}\SetKwInOut{Output}{Output}
  \SetKwFunction{ProcessPath}{ProcessPath}{}{}
  \SetKwFunction{reverse}{reverse}{}{}
  \SetKwFunction{LP}{LP}
  \SetKwFunction{LP}{LP}

  \Input{Directed Graph $G(V, E)$\\
  			Source-Sink Pair $(s_i, t_i)$ for each satisfiable commodity $i\in [k]$ \\% each $ F_i \in \mathcal{F}$\\
  			Capacity $c(u, v) ~ \forall (u, v)\in E$\\
  			Estimator $\estP: \{0,1\}^k \to \mathbb{R}_{\geq 0}$ for obtaining an $(\alpha,\beta)$-approximate sol.}
  \Output{$(\alpha,\beta)$-approximate solution to the ANF instance}
  
  \BlankLine

\COMPUTE optimal solution $\vec{\tilde{f}}$ to compact edge-flow LP (cf. Figure~\ref{fig:compact})\\

\LET $Z_i \in \{0,1\}$ be the random variable s.t. $\Prob[Z_i = 1] = \tilde{f}_i$ and $\Prob[Z_i = 0] = 1- \tilde{f}_i$ for $i\in [k]$\\

\COMPUTE $\failureestimate \gets \Expec[\estP(Z_1,...,Z_{i-1},Z_{i},...,Z_n)]$\\

\ForEach(\tcp*[f]{iterate over all commodities}){$i \in [k]$}{ 
			\eIf{$\Expec[\estP(z_1,\ldots,z_{i-1},0,Z_{i+1},...,Z_n)] < \failureestimate$}
			{
				
				\SET $z_i \gets 0$ 				\tcp*{commodity $i$ is \emph{not} routed}
			}{
				\SET $z_i \gets 1$ 			\tcp*{commodity $i$ is routed}
			}
			\UPDATE $\failureestimate \gets \Expec[\estP(z_1,\ldots,z_i,Z_{i+1},\ldots,Z_n)]$\\
}
\KwRet{solution given by $\vec{z}$: if $z_i = 1$ then $f_i = 1$ and $f_{i,e} = \tilde{f}_{i,e}  / \tilde{f}_i$, else $f_i = f_{i,e} = 0,\,\forall i \in [k]$}
  \caption{{Deterministic Approximation for the All-or-Nothing Flow Problem}}
  \label{alg:deterministic-rounding}
  \end{algorithm*}
  \vspace{-12pt}
  \end{figure}

We now introduce the following specific pessimistic estimator $\estP$ for which the above three correctness criteria (upper bound, efficiency, base case) are proven.
As before, in the following we always assume $\alpha = 1-1/m$ and $\beta = 3b\ln m/\ln\ln m$ for $m \geq 9$ and $b = 1.85$.

\vspace{-.05in}
\begin{restatable}[Pessimistic Estimator]%
{lemma}{pessimisticEstimatorForANF}
\label{lem:pessimistic-estimator}
The function $\estP$ is a pessimistic estimator for the ANF:
\label{def:pessimistic-estimator}
\allowdisplaybreaks
	\vspace{-.05in}
\begin{alignat*}{7}
	&\quad\quad\quad\quad\quad \estP(Z_1,\ldots,Z_k)  =  \estA + \sum_{(u,v) \in E} \estB,\,\,\textnormal{\it where}\\
	&%
	\estA  =  \expT[-\paramalpha  (1-\delta_{\alpha}) \exptilde]  \prodEnum[\Expec[\expT[\paramalpha Z_{i} \frac{ w_i}{w_{\max}}]]],
	\,\textnormal{\it with\ }\delta_{\alpha} = \frac{1}{m}\,,\, \exptilde = \frac{w_{LP}}{w_{\max}}\,,\, \paramalpha = \ln(1-\delta_{\alpha}); \\ %
	&\textnormal{\it and\ }%
	\estB  =  \expT[-\parambeta  (1+\delta_{\beta}) \exphat] %
	\prodEnum[\Expec[\expT[\parambeta  Z_{i} \frac{(\flowonuvbyi / \tilde{f}_i)}{\capofuv}]]],%
	\,\textnormal{\it with\ }\delta_{\beta} = \frac{3b\ln m}{\ln\ln m - 1},\, b = 1.85, \\
	& \exphat = 1,\, \parambeta = \ln(1+\delta_{\beta})\,.
\end{alignat*}
\end{restatable}

\begin{proof}
The following three properties are to be shown: (i) upper bound, (ii) efficiency, and (iii) base case (cf. Equations~\ref{eq:pess:bound} - \ref{eq:pess:base}). We first discuss properties (i) and (iii).

\sloppy
The analysis in Section~\ref{sec:randomized-algorithm} has demonstrated that the probability of obtaining an \mbox{$(\alpha,\beta)$-approximate} solution via randomized rounding is bounded from below by $1/(6\cdot m^2)$~(cf. Corollary~\ref{cor:randomized-rounding-result}). To obtain this result, a union bound argument was employed, which used probabilistic bounds on not achieving at least an $\alpha$ fraction of the optimal throughput and exceeding the capacity of each single edge by a factor of $\beta$. 

For the throughput, the Chernoff bound of Theorem~\ref{thm:chernoff-lower} is applied, while for each edge's capacity violation, the Chernoff bound of Theorem~\ref{thm:chernoff-upper} is used. The pessimistic estimators $\estAF$ and $\estBF$ are a direct result of these respective theorems:
\begin{itemize}
\item $\estAF$ is obtained from the application of the Chernoff bound of Theorem~\ref{thm:chernoff-lower} within the proof of Lemma~\ref{lem:random-weight}. Specifically, the application of the Chernoff bound in Lemma~\ref{lem:random-weight} yields the following --- restated over the variables $Z_i$ --- with the parameters $\delta_{\alpha}$, $\paramalpha$ and $\exptilde$ as specified above:
\[
\Prob[
\sumEnum[w_l \cdot Z_l] < \alpha\cdot w_{LP} %
] \overset{}{\leq} \expT[-\paramalpha \cdot (1-\delta)\cdot \exptilde] \cdot {\prod_{i \in[k]}\Expec[\expT[\paramalpha \cdot Z_{i} \cdot w_{i}/ w_{\max}]]} \overset{}{\leq}  e^{-1/(2\cdot m^2)} 
\]

 The  middle expression directly yields the pessimistic estimator for the throughput.
\item $\estBF$ is analogously obtained from the application of the Chernoff bound of Theorem~\ref{thm:chernoff-upper} in the Lemma~\ref{lem:overflow-prob} for each edge $(u,v) \in E$. Specifically, for a single edge $(u,v)$, the following is obtained when using the constants defined above:
\[
\Prob[\sum_{i\in [k]} f_{i,(u,v)}> \beta \cdot c(u,v)
] \leq 
\expT[-\parambeta \cdot (1+\delta_{\beta}) \cdot \exphat] \cdot \prod_i \Expec[\expT[\parambeta \cdot Z_{i} \cdot \flowonuvbyi / \capofuv]] \leq 1/(6\cdot m^2)
\]
Again, the middle expression is used to obtain the pessimistic estimator $\estBF$ for the specific edge $(u,v) \in E$.
\end{itemize}
Revisiting the union bound argument, we obtain that $\estP$ indeed yields an upper bound on the failure probability to construct an $(\alpha,\beta)$-approximate solution, and that initially $\Expec[\estP(Z_1,\ldots,Z_k)] \leq 1 - 1/(6\cdot m^2) < 1$ holds for $m \geq 9$. This shows that properties (i) and (iii) are satisfied.

Considering the efficiency property (ii), we note the following. 
Both $\estAF$ and $\estBF$ consist of products, where expectations for different commodities can be computed independently. Due to the binary nature of the variables $Z_i$, these expectations can be computed in constant time.
\end{proof}

Given the above outlined intuition of the derandomization process and the correctness of the pessimistic estimator due to Lemma~\ref{lem:pessimistic-estimator}, the following main theorem of this section is obtained.

\begin{theorem}
Using $\estP$ as a pessimistic estimator, Algorithm~\ref{alg:deterministic-rounding} is a deterministic \mbox{$(\alpha,\beta)$-approximation} for the ANF problem with $\alpha=1-1/m$ and $\beta = 3b\ln m/\ln\ln m\,$, with $b = 1.85
$,  for $m \geq 9$.
\end{theorem}

\section{Potential Problem Extensions}
\label{sec:extensions}
The packing formulation for the ANF was introduced in Section~\ref{sec:anf-packing} together with a simple randomized rounding algorithm. Besides the practical tractability established in Section~\ref{sec:mwu}, the proposed packing framework for the ANF has further advantages. Specifically, it can be easily adapted to cater for problem extension such as when flows are restricted to $k$-splittable flows, must obey fault-tolerance criteria, or are restricted to shortest paths. 

In the following we describe some of these extensions and how the packing formulation may be adapted together with the separation procedure. Notably, some problem extension allow for compact LP formulations, however, casting the problems in terms of the packing formulation is generally less complex and therefore helps in establishing whether a problem extension can be efficiently approximated in the first place.

Henceforth, our
goal is to solve the maximum throughput problem in the
all-or-nothing model while restricting the nature of flows that are
allowed for each commodity. The ANF allows flow for each commodity to be
split in arbitrary ways while unsplittable flow requires all the flow
for a commodity to use a single path.  However, there are several
intermediate settings of interest, and other constraints, that occur in practice.
Recall that in setting up the formulation in
Figure~\ref{fig:lp-flow}, $\mathcal{F}_i$ for each
pair $i$ is the set of valid $s_i$-$t_i$ flows in $G$.  This
is a large implicit set, and the way we solve the LP relaxation is via the
separation oracle. The separation oracle corresponds to finding a minimum-cost
flow from $\mathcal{F}_i$ given some edge lengths/costs.  The MWU algorithm can be
viewed as an efficient, albeit approximate, way to solve the large
implicit LP relaxation via the separation oracle. Moreover, once the
LP is solved, the randomized rounding step picks one of the flows per
commodity. This flexibility allows us to solve the LP and round even
when $\mathcal{F}_i$ is restricted in some fashion.  We outline
a few extensions that can be addressed via this framework.

\smallskip \noindent
\textbf{Integer flows:} Recall that in ANF we allow splittable
flows. However in some settings it is useful to have flow for each
commodity on each edge to be integer valued; here we assume that
$d_i$ is an integer for each $i$. In order to handle this we can
set $\mathcal{F}_i$ to be the set of all integer $s_i$-$t_i$ flows.
Now the min-cost flow routine needs to find an integer flow  between
$s_i$ and $t_i$ of $d_i$ units. This is easy to ensure since there
always exists an integer valued min-cost flow as long as the demands and the edge capacities are all
integer valued. We reduce
each $c(e)$ to $\lfloor c(e)\rfloor$ without loss of generality.

\smallskip \noindent
\textbf{Splitting into a small number of paths:} In some applications it is important that the flow for each pair is not split by too much. How do we quantify this? One way is to consider $h$-splittable flows where $h$ is an integer parameter. This means that flow for each pair can be decomposed into at most $h$ paths. When $h=1$ we obtain unsplittable flow and if we set $h=|E|$ we obtain ANF. 
We can handle the special case
where each of the $h$ flow paths has to be used to send the same
amount of flow which is $d_i/h$. For this purpose we define
$\mathcal{F}_i$ to be the set of all such flows.  To compute a
min-cost flow in $\mathcal{F}_i$ we simply need to find a min-cost
flow of $h$ units from $s_i$ to $t_i$ in the graph with capacities
adjusted as follows:  for each edge $e$ with capacity $c(e)$ we
change it to $\lfloor h c(e)/d_i \rfloor$. Baeier et al.\ \cite{BaierKS05}
considered this maximum throughput problem, however, they only considered uni-criteria approximation algorithms and provided a reduction to the unsplittable case; the approximation ratios that one can obtain without violating capacities are very poor while our focus here is on bicriteria approximation that achieve close to optimum throughput.

\smallskip \noindent
\textbf{Fault-tolerance and routing along disjoint paths:}
In some settings the flow for a pair $(s_i,t_i)$ needs to be
fault-tolerant to edge and/or node failures. There are several ways
this is handled in the networking literature. One common approach is
to send the flow for each commodity along $h$ disjoint paths, each
carrying $d_i$ units. This can be handled by an approach
very similar to the preceding paragraph where we compute min-cost
flow on $h$ disjoint paths; note that in the preceding paragraph the $h$ paths could share edges. Another approach to fault-tolerance is to use what are called $h$-route flows \cite{Kishimoto96,AggarwalO}. One can find a min-cost  $h$-route flow in polynomial time \cite{AggarwalO}. Hence, one can also use the framework to maximize throughput while each routed pair uses an $h$-route flow. %\redcomment{AR: PLs check if what I added in red is indeed true} \chandra{I removed your min-cost. The min-cost problem is required to solve the dual. Recall that the framework is to maximize the weighted throughput.}

\smallskip \noindent
\textbf{Using few edges or short paths:} We now
consider the setting when the flow for a commodity is required to use a
small number of edges or the flow has to be routed along paths with small
number of hops. These constraints not only arise in practice but also
help improve the theoretical bounds on congestion. One can show that
if each flow uses only $d$ edges then the bicriteria approximation can
be improved; the congestion required for a constant factor
approximation becomes $O(\log d/\log \log d)$ rather than
$O(\log m/\log \log m)$; for single paths the analysis can be seen from \cite{CCGK07} and we can generalize it to our setting. Suppose we wish to route flow for each
commodity whose support consists only of some given number $h$ of
edges. As above we need to solve for min-cost $s_i$-$t_i$ flow that
satisfies this extra constraint. However this additional constraint is
no longer so easy to solve and in some cases can be NP-Hard.  However,
if one allows for a constant factor relaxation for the number of edges
$h$, and an additional constant factor in the edge congestion one can
address this more complex constraint by using 
linear programming based ideas (see~\cite{ChekuriI18} for an example).

\section{Simulation Results}
\label{sec:sim}

In this section we study the performance of our approximation algorithms for the ANF problem on real-world networks. 
Our proof-of-concept computational evaluation is meant to provide general guidelines about the relative efficacy of the algorithms in terms of the {\em achieved throughput approximation factor $\alpha$} and the {\em edge capacity violation ratio $\beta$}.
The achieved throughput approximation ratio is taken as the solution obtained by the run divided by the optimal LP solution (which is a lower bound on the exact approximation ratio based on the optimal IP solution rather that its LP relaxation). Notably, due to the bi-criteria nature of our approximations with solutions being allowed to exceed edge capacities (by at most a factor of $\beta$), solutions may yield empirical throughput approximation factors of $\alpha > 1$.

Beyond analyzing the performance of our randomized rounding and derandomized algorithms, we also investigate the impact of varying the methodology by which the LP is solved. Specifically, we study the performance of solving the compact LP formulation directly, of solving the multiplicative weight update algorithm (MWU), and of solving the MWU-based Permutation Routing (PR) heuristic described below. While the runtime of our prototypical MWU implementation generally exceeds the runtime of solving the compact LP formulation using a commercial solver, our MWU implementation serves as a proof-of-concept of its practical applicability and will also enable the extensions outlined in Section~\ref{sec:extensions}, which depend on the packing formulation. In addition, we remark that MWU may be useful for larger networks in practice (larger than the ones considered here), as it does not suffer from the same space complexity limitations as  solving the compact LP via standard LP solvers.

Note that the simulation results for the current state-of-the-art algorithm for constant-throughput approximations for the ANF problem~\cite{Liu19} --- originally designed in to handle uniform demands, edge capacities and weights --- have been reproduced in this paper when running the randomized rounding algorithm with the compact edge-flow LP, since this algorithm is in essence the same as the algorithm in~\cite{Liu19}, adapted here to handle non-uniform demands, edge capacities and weights (in addition to some fine tuning optimizations). Our theoretical approximation results in this paper actually also validate the simulation results in~\cite{Liu19}, since the simulations in~\cite{Liu19} already suggested that the edge capacity violations incurred by randomized rounding based on the compact edge-flow LP were logarithmic (and not polynomial as the theoretical guarantees of~\cite{Liu19} suggested).

\subsection{Permutation Routing Heuristic}

Without proper optimization, the runtime of the MWU algorithm can be slow due to the computation of $k$ minimum-cost flows as a separate procedure in each iteration. 
As a practical solution, we introduce a heuristic based on Algorithm~\ref{alg:MWU} that provides a significant reduction in computational cost, while still yielding solutions comparable to those by MWU in practice. We refer to this as the {\em Permutation Routing (PR)} algorithm. In the following we outline how this new algorithm differs from the original MWU algorithm and we refer the reader to Algorithm~\ref{alg:permutationrouting} for the complete pseudocode description.

\begin{algorithm}[h!]
\caption{Permutation Routing Algorithm for the Multi-Commodity ANF Problem}

\begin{algorithmic}[1]
    \Statex{\textbf{Inputs:} $\gamma \in \mathbb{R}^{+}$, Directed Graph G(V,E), $c:E \to \mathbb{R}^{+}$, a set $S$ of $k$ pairs of commodities $(s_{i},t_{i})$ each with demand $d_{i}$, weight $w_i$, an estimate $Est$ of the optimal fractional ANF solution for $(G,S)$  }
    \State{Initialize an empty flow $f(e)\leftarrow 0,\forall e\in E$}
    \State{Set $\eta\leftarrow \frac{\ln{|E|}}{\gamma}$}
    \State{Set $r\leftarrow \frac{\ln{|E|}}{\gamma^2}$}
    \State{Let $f_{i,e}\leftarrow 0$ be the fractional flow assignment for commodity $i$ on edges $e$ for all $i\in S,e\in E$}
    \State{Define edge costs $\ell(e)=1,\forall e\in E$ }
    \State{Make $r$ copies the $k$ commodities of $S$ and let $A$ be a list of these $rk$ commodities}
    \State{Let $B$ be a random permutation of $A$ }
    \For{each commodity copy $j$ in $B$}
        \State{Let commodity $j$ in $B$ correspond to original demand $(s_i,t_i)$}
        \State{Compute min-cost flow of $d_i$ units from $s_i$ to $t_i$ with edge costs defined by $\ell$ and obtain flow assignment $f'$ and solution cost $\rho=\sum_{e\in E}\ell(e)f'(e)$} 
        \State{Compute $\tau=\sum_{e\in E}\ell(e)c(e)$}
        \If{$\frac{w_j}{\rho}\geq\frac{Est}{\tau}$ and none of the following updates cause an edge capacity violation}
            \For{ each edge $e\in E$}
                \State{Update $f(e)\leftarrow f(e)+\frac{f'(e)}{r}$}
                \State{Update $f_{i,e}\leftarrow f_{i,e}+\frac{1}{r}$} 
                \State{Update $\ell(e)\leftarrow\exp{\bigg(\frac{\eta\cdot f(e)}{c(e)}\bigg)}$}
            \EndFor
        \EndIf
        
    \EndFor
    \State{Return $f(e),f_{i,e}$ for all $i\in S, e\in E$}
\end{algorithmic}
\caption{Permutation Routing Algorithm}
\label{alg:permutationrouting}
\end{algorithm}

 Our algorithm is motivated by theoretical algorithms for maximum throughput packing problems in the \emph{online} arrival model and the \emph{random arrival order} models. It is known that for packing problems, in the random arrival model, one can obtain arbitrarily good performance compared to the offline optimal solution if the resource requirements of the arriving items (these correspond to flows in our setting) are sufficiently small when compared to the capacities \cite{AD14,GuptaM,Kesselheimetal}. The analytical ideas are related to online learning and MWU. %

We develop our heuristic as follows: Recall that we are seeking a fractional solution.
We take each commodity pair $i$ with demand $d_i$ and split it into $r$ ``copies," each with a demand of $d_i/r$. Here $r$ is a sufficiently large parameter to ensure the property that $d_i/r$ is ``small" compared to the capacities. From the MWU analysis, and also the analysis in random arrival order models, one sees that $r = \Omega(\ln m/\gamma^2)$ suffices. Given the $k$ original pairs, we create $k\cdot r$ total pairs from the copies. We now randomly permute these pairs and consider them by one-by-one.  When considering a pair, the algorithm evaluates the ``goodness'' of the pair in a fashion very similar to that of the MWU algorithm. It maintains a length for each edge that is exponential in its current loads, and computes a minimum cost flow for the current pair (note that the pair's demand is only a $1/r$ fraction of its original demand); it accepts this pair if the cost of the flow is favorable compared to an estimate of the optimum solution. If it accepts the pair, it routes its entire demand (which is the $1/r$'th fraction of the original demand). Otherwise this pair is rejected and never considered again. Thus the total number of minimum cost flow computations is $O(k\cdot r)$ when compared to $O(k\cdot m \cdot\log m/\gamma^2)$ in the MWU algorithm. As mentioned above, a worst-case theoretical analysis requires $r = \Omega(\log m/\gamma^2)$ to guarantee a $(1-\gamma)$-approximation, however, in practice a smaller value of $r$ can be chosen. Note that an original pair $(s_i,t_i)$ with demand $d_i$ is routed to a fraction $r_i/r$ where $r_i$ is the number of copies of $i$ that are admitted by the random permutation algorithm. The algorithm requires an estimate of the optimum solution which can be obtained via binary search or other methods.

\subsection{Methodology}
\label{sec:experimental-design}

We now describe the 
problem instances and the implementations of our approximation algorithms. 

\paragraph{\bf Problem Instances}
Following \cite{Liu19}, we study real-world networks together with corresponding real-world source-sink pairs obtained from the survivable network design library (SNDlib)~\cite{sndlib}.
We randomly perturb the uniform weights, demands and edge capacities of the chosen networks  to test our algorithms' ability to accommodate variable weights and demands on networks with varying edge capacities.
Due to this choice, we find that only a fraction of the given commodities can be concurrently satisfied.
%Moreover, the choice ensures that few (if any) commodities can be fully routed through a single path. \redcomment{AR: I think the following is redundant and propose to cut from text, unless anyone explains to me why we shouldn't: "without over-saturating the network, but still allows for a non-trivial fraction of the flows to be routed."} 
Our choice of networks from the SNDlib is given in Table~\ref{table:networks}, covering several general scenarios, e.g.~a small network with  large number of commodities, or a dense network with large number of commodities. We chose independent uniform random edge capacities from 20 to 60, commodity demands from 25 to 75, and commodity weights from 1 to 10 (the benchmark SNDlib data has all edge capacities at 40, demands at 50, and weights at 1).

\begin{table}[bth]
\begin{tabular}{lllll}

\textbf{Network } & \textbf{Vertices} & \textbf{ Edges} & \textbf{ Commodities } & \textbf{General Description}                    \\ \hline
Atlanta               & 15                     & 44                  & 210                           & Small network, high commodity count \\
Germany50             & 50                     & 176                  & 662                           & Sparse network,  high commodity count \\
Di-yuan               & 11                     & 84                  & 22                            & Dense network,  low commodity count  \\
Dfn-gwin              & 11                     & 94                  & 110                           & Dense network,  high commodity count  
\end{tabular}
\caption{List of studied adapted instances from SNDlib~\cite{sndlib}}
\label{table:networks}
\vspace{-.3in}
\end{table}

\paragraph{\bf Algorithms}
We have implemented both the randomized and derandomized rounding algorithms detailed in Sections~\ref{sec:randomized-algorithm} and~\ref{sec:derandomization}. We solve the compact formulation via CPLEX V12.10.0 and approximately solve the packing LP via the MWU algorithm or the faster permutation routing heuristic.
We choose $\eps=\frac{1}{9}$ and $b=1.85$ in Algorithms~\ref{alg:randomrounding} and \ref{alg:deterministic-rounding}, implying a target throughput approximation factor of $\alpha \geq 1-\eps=\frac{8}{9}$ and
target edge capacity violation ratio of $ \beta \leq {3b\ln{m}}/{\ln{\ln{m}}}={5.55\ln{m}}/{\ln{\ln{m}}} $, where $m$ is the number of network edges, for the algorithms.
More specifically, for the Atlanta and Germany50 networks, we target edge capacity violations $\beta\leq 15.78$ and $17.47$, respectively. 

We define an experiment as the execution of a higher level algorithm (either randomized or derandomized rounding) in concert with an LP-solving subroutine (CPLEX for compact LP or our MWU and PR implementations) on a particular network. 
For an experiment that includes randomized rounding, we execute this algorithm 10 times to obtain a total of 10 different samples per experiment. 
For each of these 10 executions, 100 rounds of rounding are recorded and of these rounded solutions, we report on the solution of highest throughput whose capacity violations lie below our theoretical bounds.
We consider three different $\gamma$ values, namely $0.15$, $0.2$, and $0.3$, to study performance vs. runtime trade offs of the MWU algorithm and the PR heuristic. 
Due to the slow convergence of MWU, we introduce speed-up mechanisms where (i) during any iteration, if the post-update smallest mincost flow solution is not at least 50 percent larger than the pre-update smallest min-cost flow solution, then we do not recompute this in the subsequent iteration, and (ii) the maximum number of iterations is capped at 10k.
\begin{figure}[h!]
	\vspace{-.2in}
	\centering
	\includegraphics[width=\textwidth]{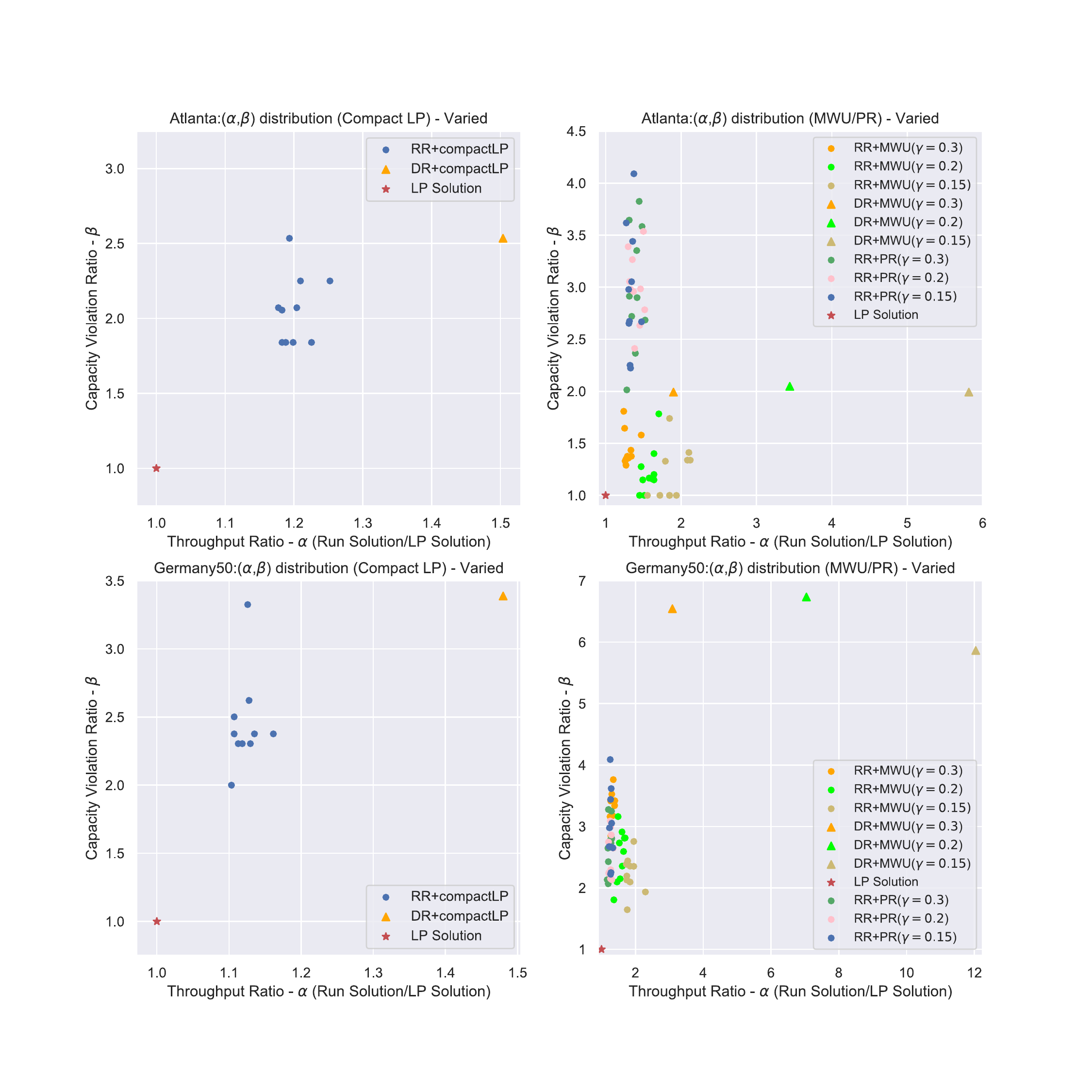}
	\label{fig:mainfigvaried}
	\centering
	\vspace{-.5in}
	\includegraphics[width=\textwidth]{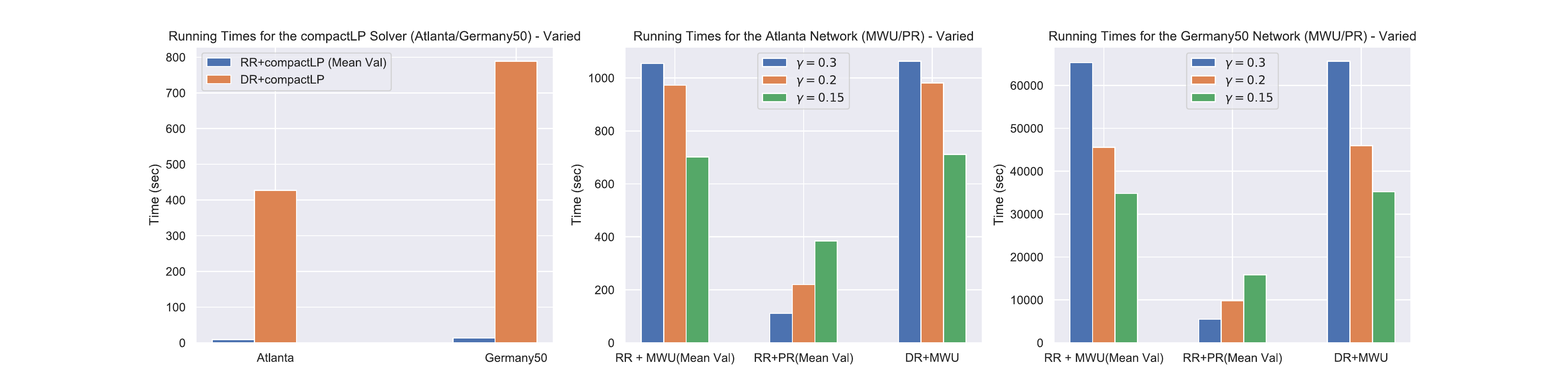}
	\caption{Experimental results  on the Atlanta (top) and Germany50 (middle) networks and their runtimes (bottom); RR refers to the randomized and DR to the derandomized rounding algorithm. Edge capacities vary uniformly at random from 20 to 60, commodity demands vary uniformly at random from 25 to 75, and commodity weights vary uniformly at random from 1 to 10. }
	\label{fig:runtimefig}
\end{figure}

\subsection{Experimental Results}

In this section we report on our computational results. 
We first focus our attention on the performance of the Atlanta and Germany50 networks 
but will then also discuss two smaller networks.
We report results in terms of the achieved throughput factor $\alpha$, edge capacity violation factor $\beta$, and the wall-clock running times. 

Our experiments are summarized visually in Figure~\ref{fig:runtimefig}, and we will refer to this figure for the remainder of this section. The qualitative plots at the top and the middle show the empirical throughput and edge capacity violation ratios obtained by executions of the various algorithms. Note that we report on 10 data points when applying randomized rounding in contrast to the single data point for the derandomized algorithm. For reference, we include a red star data point indicating the optimal
LP solution. 

We see for both the Atlanta and Germany50 networks that the compact LP combined with both the randomized and derandomized algorithms produces $\beta$ values that are well within our established theoretical bounds, although we see noticeably larger values of $\alpha$ and marginally larger values of $\beta$ when the derandomized rounding is employed.
The values of $\alpha$ and $\beta$ obtained from MWU are similarly concentrated around their means. Interestingly, in combination with the MWU algorithm, the deterministic rounding shows for Germany50 a significant increase in edge capacity violations while also achieving a much higher throughput. For the deterministic rounding of MWU, we observe that $\alpha$ increases as $\gamma$ decreases under roughly constant capacity violations.
With respect to the permutation routing subroutines, we observe more variance over the parameter space, and we typically see much higher capacity violation without a significant gain in throughput.

Regarding the runtimes, we remark that the solving the compact LP using CPLEX is in general significantly faster than our MWU or PR implementations. Furthermore, the randomized rounding algorithm significantly outperforms our efficient deterministic rounding algorithm. We believe this to be mainly due to our naive implementation of the pessimistic estimators, that does not cache intermediate results.

Regarding the performance of the permutation routing algorithm, we note the following: For both networks we see a runtime decrease of a factor of at least half compared to the vanilla MWU algorithm while slightly compromising on the generally less favorable higher capacity violations. 
For the MWU algorithm, we expect in general to see the runtime increase as $\gamma$ decreases, however, due to our speed-up mechanism, the opposite may be true.
This is attributed to the fact that with a smaller $\gamma$, the increase in flow is likewise smaller, and thus the updates to edge costs are smaller, implying that the threshold for skipping min-cost flow calculations is met more often.
Thus, runtime is reduced for smaller values of $\gamma$, though at the expense of worse throughput approximations. 

Figure~\ref{fig:runtimefiguniform} summarizes the results on networks Atlanta and Germany50 under the default uniform weights, demands and edge capacities given by \cite{sndlib}.
This figure also replicates experiments from \cite{Liu19}, though here we additionally test Algorithms~\ref{alg:MWU},\ref{alg:permutationrouting} in conjunction with Algorithms~\ref{alg:randomrounding},\ref{alg:deterministic-rounding}.
In Figure~\ref{fig:appendixfigvaried}, we present experimental results on two additional networks from \cite{sndlib}, namely DFN-GWIN and DI-YUAN.
In these experiments, we randomly perturb the weights, demands and edge capacities in the same manner given in Figure~\ref{fig:runtimefig}. 
Lastly, in Figure~\ref{fig:appendixfig}, we again present experimental results on DFN-GWIN and DI-YUAN, though we test uniform commodity weights, demands and edge capacities.

\begin{figure}[H]
    \vspace{-.2in}
	\centering
	\includegraphics[width=\textwidth]{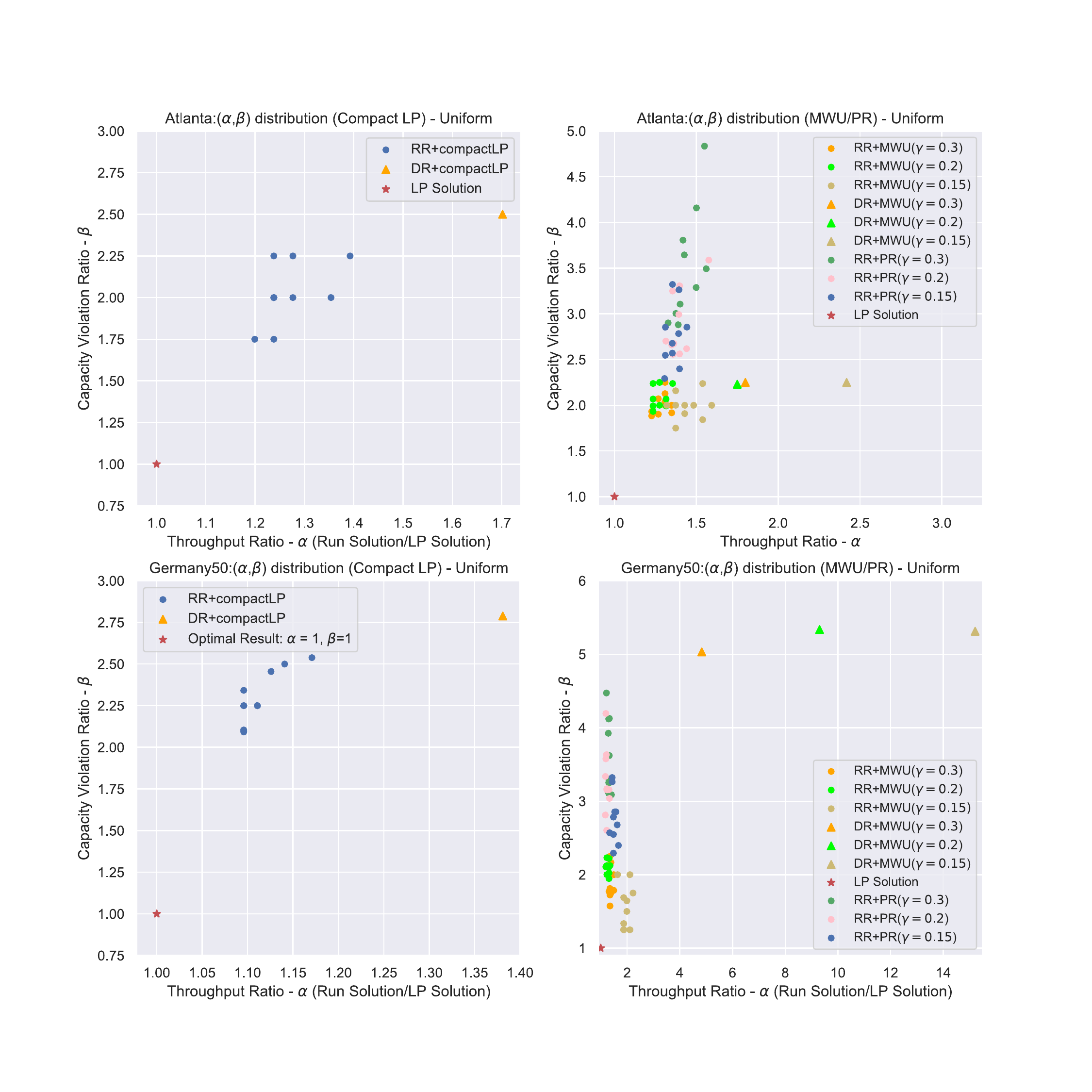}
	\label{fig:mainfig}
	\centering
	\includegraphics[width=\textwidth]{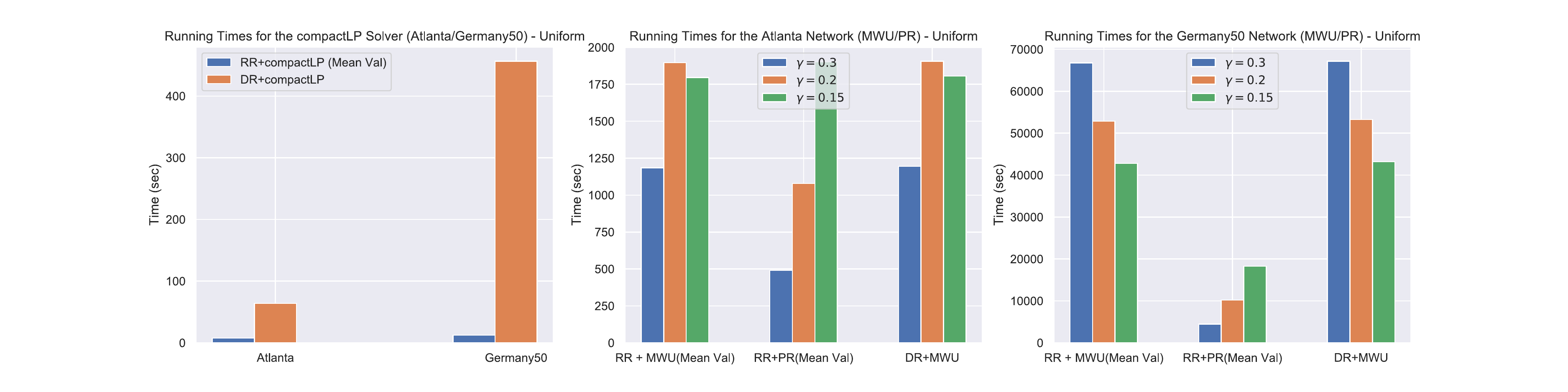}
	\caption{Experimental results  on the Atlanta (top) and Germany50 (middle) networks and their runtimes (bottom); RR refers to the randomized and DR to the derandomized rounding algorithm. Edge capacities, commodity weights and commodity demands are fixed to 40, 1 and 50, respectively per \cite{sndlib}.}
	\label{fig:runtimefiguniform}
\end{figure}

\begin{figure}[H]
    \vspace{-.2in}
	\centering
	\includegraphics[width=\textwidth]{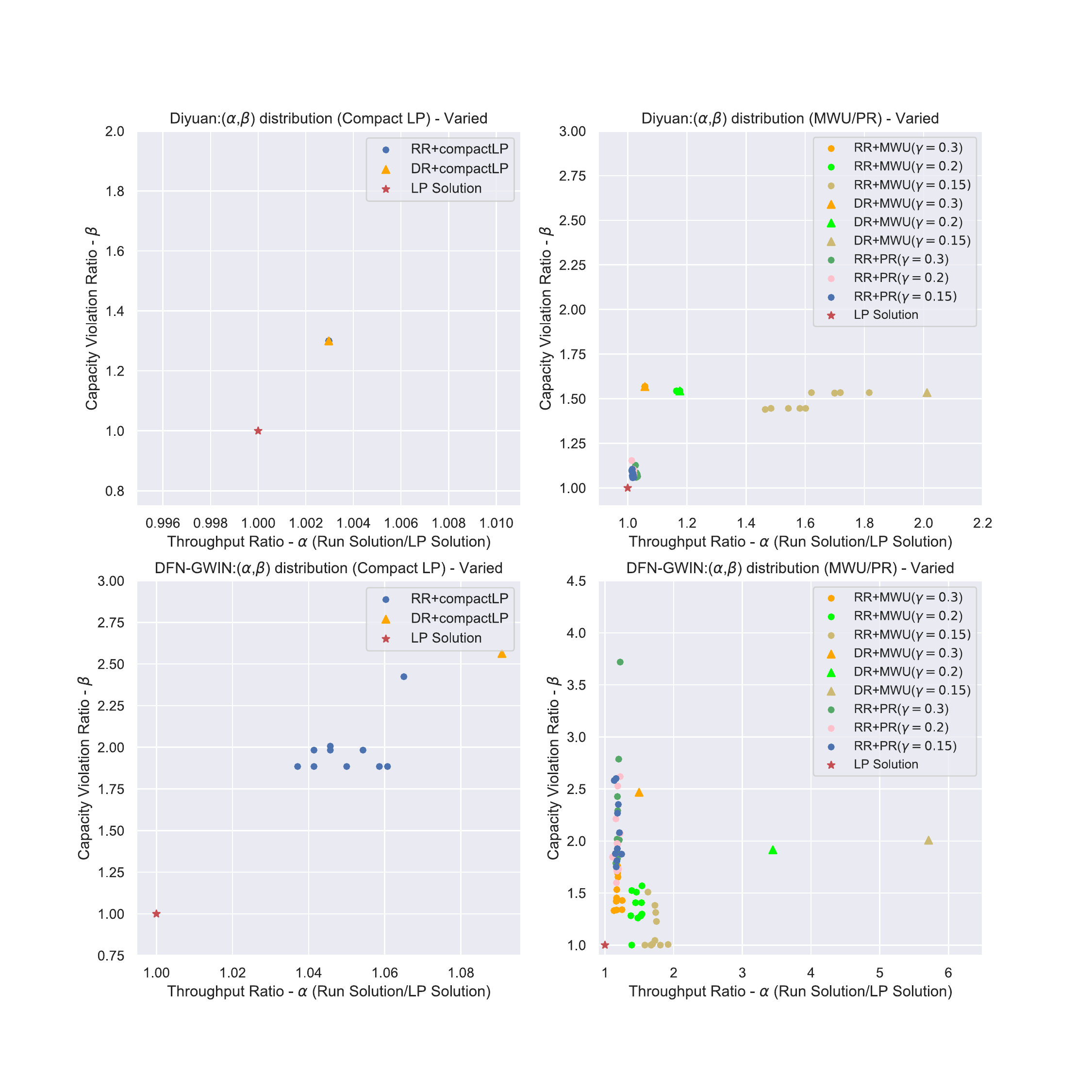}
	\centering
	\includegraphics[width=\textwidth]{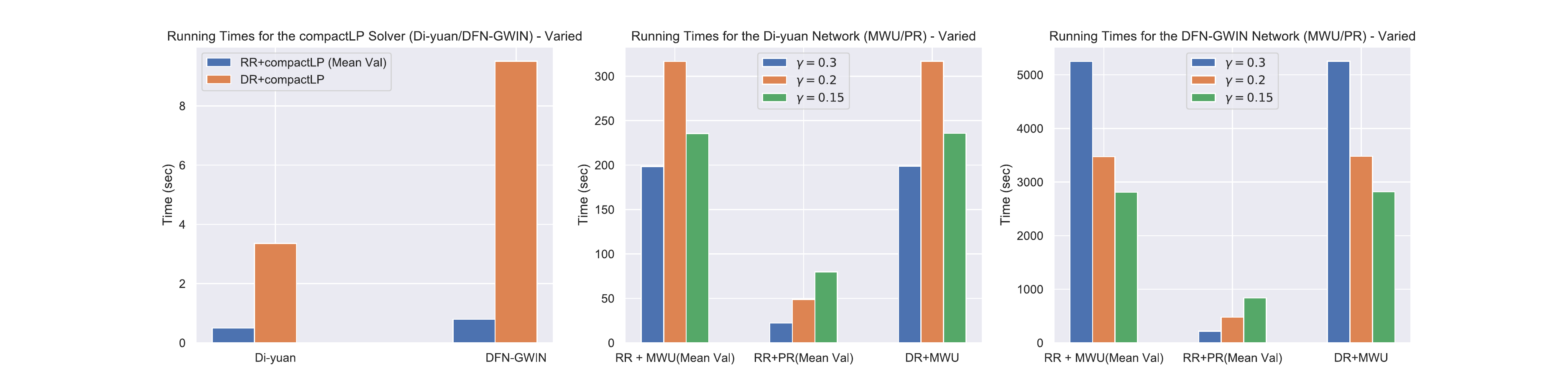}
    \caption{Experimental Results: $\alpha \times \beta$ plots and running times. Note that our theory maintains that $\beta\le 16.52$ and $\beta \le 16.66$ for the Di-yuan and DFN-GWIN networks, respectively. Edge capacities vary uniformly at random from 20 to 60, commodity demands vary uniformly at random from 25 to 75, and commodity weights vary uniformaly at random from 1 to 10.  }
    \label{fig:appendixfigvaried}
\end{figure}

\begin{figure}[H]
    \vspace{-.2in}
	\centering
	\includegraphics[width=\textwidth]{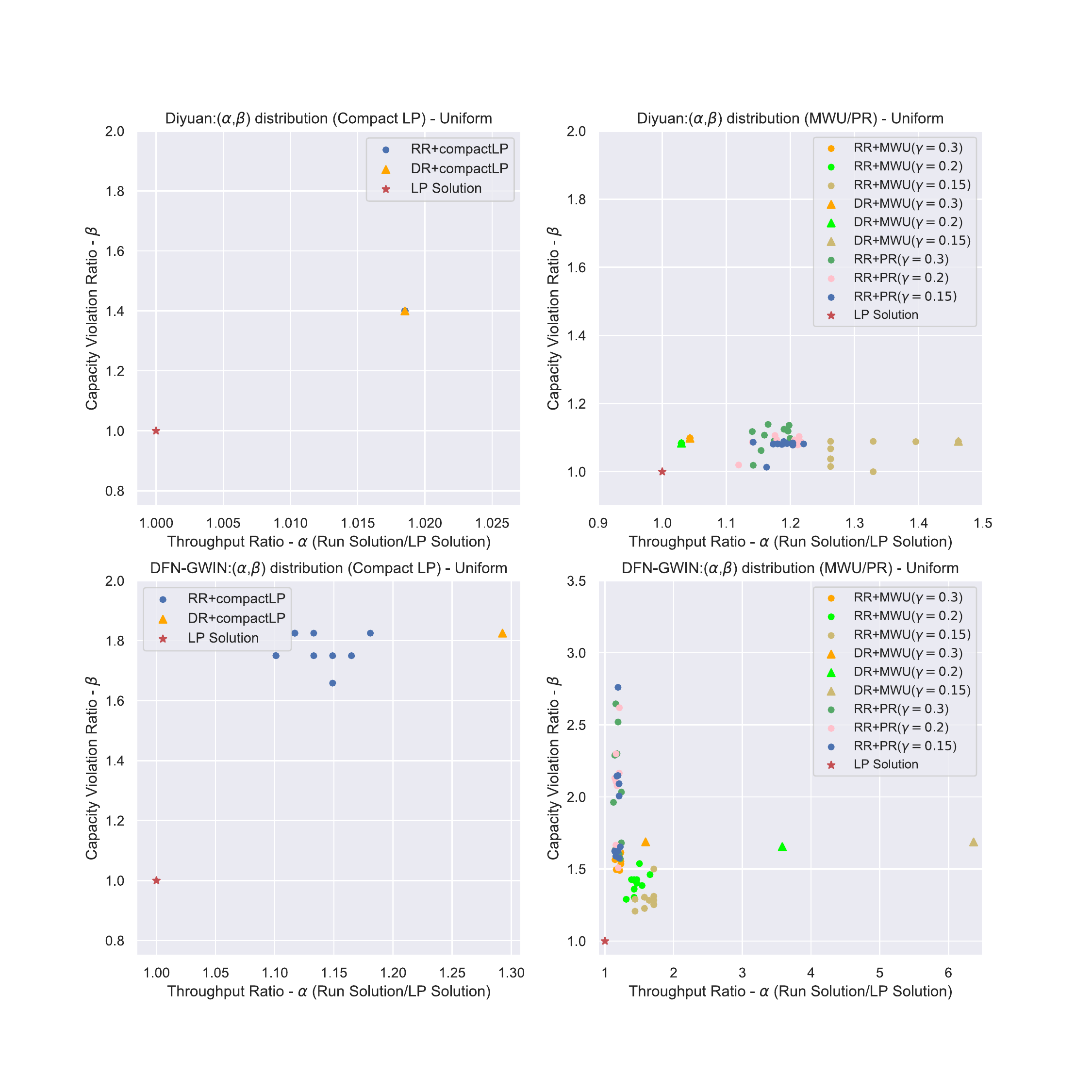}
	\centering
	\includegraphics[width=\textwidth]{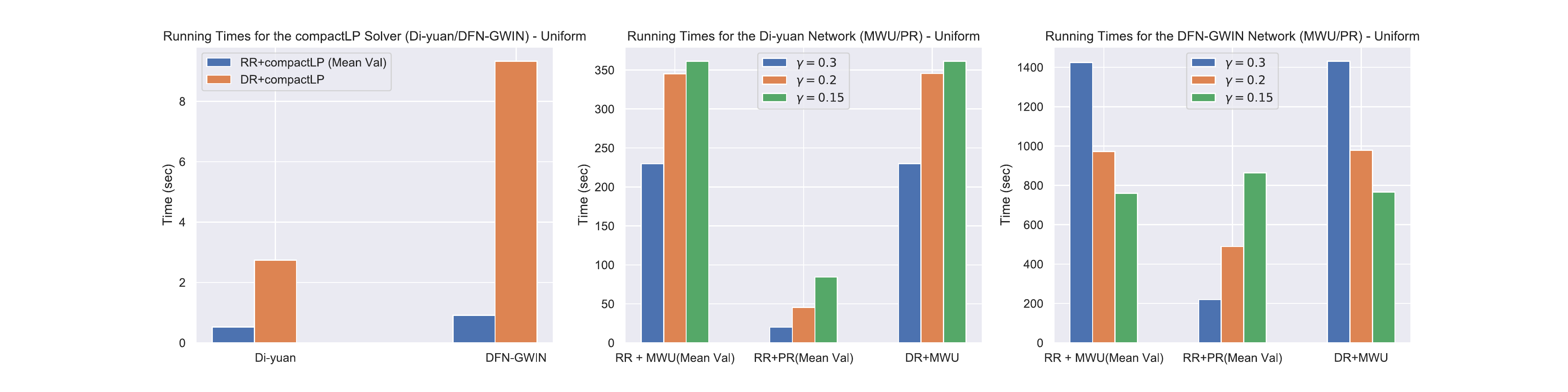}
    \caption{Experimental Results: $\alpha \times \beta$ plots and running times. Note that our theory maintains that $\beta\le 16.52$ and $\beta \le 16.66$ for the Di-yuan and DFN-GWIN networks, respectively. Network capacities, commodity weights and commodity demands are fixed to 40,1 and 50, respectively per \cite{sndlib}.}
    \label{fig:appendixfig}
\end{figure}

Concluding, we note that we see our results as a first step towards efficiently approximating the ANF and its potential extensions. While randomized rounding wins in terms of runtime, the deterministic rounding generally achieves slightly higher throughputs. Furthermore, while solving the compact LP is shown to be much quicker in practice, the proposed MWU algorithm will render tackling the problem extensions of Section~\ref{sec:extensions} tractable and our proposed permutation routing heuristic can in practice substantially reduce runtimes.
With respect to space, it is least efficient to solve the compact LP directly.
In fact, it may be impossible to do so if the network is large enough.
The MWU and Permutation Routing algorithms, on the other hand, rely on repeated and discarded computations of single commodity min-cost flow, whose corresponding LPs exhibit much fewer constraints than multicommodity ANF on the same network.

\section{Conclusion}
\label{sec:concl}

We presented a novel and significantly improved 
approximation of the maximum throughput routing problem
for all-or-nothing multiple commodities with arbitrary demands.
To this end, we derived formal bi-criteria  approximation bounds and 
presented a proof of concept on
efficient implementations of our algorithms in practice.
We also showed that our packing framework is very flexible and may hence
be of interest beyond the specific model considered in this paper,
and apply, e.g., to scenarios where flows should only be split
into a small number of paths or use few edges.

In future research, it would be interesting to develop improved rounding
approaches, e.g., using resampling ideas from the Lovasz-Local-Lemma, to explore the additional 
applications introduced by our packing framework, as well as to study  
opportunities for algorithm engineering, further improving the performance
of our algorithms in practice. 

\bibliographystyle{plainurl} %
\bibliography{ref,anf}

\appendix
\section{Alteration Approach}
\label{sec:alteration}
We describe an alternative rounding
approach to the one presented in Section~\ref{sec:anf-packing} that has some advantages in certain settings and gives a
better tradeoff in terms of repetitions.  The algorithm consists of
two phases. In the first phase, for each pair $i$ we randomly pick at
most one flow as described above (note that we may not pick any flow
in which case we think of it as an empty flow). For pair $i$, let $g_i$
be the chosen flow.  The second phase is an alteration phase to ensure
that the constraints are not violated more than the desired amount.
We order the pairs arbitrarily (without loss of generality from $1$ to
$k$) and consider them one by one. When considering pair $i$ we try to
route $i$ via flow $g_i$ if it is not empty.  If adding the $g_i$ to
the already chosen flows for pairs $1$ to $i-1$ does not violate any
edge capacity by more than a factor $(1 + 3b \log m/\log \log m)$, we add $i$ to
the routed pairs, otherwise we discard pair $i$. Note that a pair $i$
that was chosen in the first phase may get discarded in this second
\emph{alteration} step. Let $S$ be the random set of routed
pairs. From the construction it is clear that we can route pairs in
$S$ without violating any edge's capacity by a factor larger than $(1 +
3b \log m/\log \log m)$. Note that unlike the basic randomized
rounding algorithm we have a deterministic guarantee on this property.
Now we lower bound the expected weight of $S$.

\begin{lemma}
  Let $S$ be set of pairs routed at the end of the alteration phase.
  Then $E[w(S)] \ge (1-1/m^{\Omega(b)}) \lpval$.
  Moreover, if $y$ is a fractional solution sucht that $\lpval \ge c \lpopt$,
  for some $c \le 1$ then with probability at least $\frac{\eps -
    m^{-\Omega(b)}}{c(1+3b \ln m/\ln m)}$, $w(S) \ge (1-\eps) \lpval$.
\end{lemma}
\begin{proof}
  Consider a pair $i$. Let $Y_i$ be a binary random variable that is
  $1$ if  a non-empty flow is chosen in the random rounding
  stage. Let $Z_i$ be binary random variable that is $1$ if $i \in S$,
  that is, if $i$ is routed after the alteration phase.  We have
  $w(S) = \sum_i w_i Z_i$ and hence by linearity of expectation we
  have $E[w(S)] = \sum_i w_i \Pr[Z_i = 1]$.

  We now lower bound $\Pr[Z_i=1]$. We observe that
  $\Pr[Y_i = 1] = \sum_{f \in \mathcal{F}_i} y(f)$ by the random
  choice in the first step. We have
  $\Pr[Z_i=1] = \Pr[Y_i=1] (1 - \Pr[Z_i = 0\mid Y_i =1])$.
  The quantity
  $\Pr[Z_i = 0\mid Y_i = 1]$ is the probability that pair $i$ is
  rejected in second phase of the algorithm conditioned on the event
  that it is chosen in the first stage. Pair $i$ is rejected
  \emph{only if} there is some edge $e$ such that the total flow on
  $e$ is more than $(3b \log m/\log \log m) c(e)$ from the flows chosen
  in the first step of the algorithm; from
  Lemma~\ref{lem:overflow-prob} this probability is at most
  $1/m^{\Omega(b)}$. Thus $\Pr[Z_i = 0 \mid Y_i =  1] \le
  1/m^{\Omega(b)}$
  and hence $\Pr[Z_i = 1] \ge (1-1/m^{\Omega(b)})\sum_{f \in
    \mathcal{F}_i} y(f)$.
  Since $E[w(S)] = \sum_i w_i \Pr[Z_i = 1]$, we see that
  $E[w(S)] \ge (1-1/m^{\Omega(b)}) \lpval$.

  We now argue the second part. Let $\lpopt$ denote the value of an
  optimum solution to the LP relaxation. Consider the random variable
  $w(S)$.  We claim that $w(S) \le (1 + 3b \log m/\log \log m) \lpopt$
  deterministically. To see this recall that $S$ admits a routing that
  satisfies the capacity constraints to within a congestion of
  $(1 + 3b \log m/\log \log m)$. Therefore, by scaling down the routing
  of each demand in $S$ by a $(1 + 3b \log m/\log \log m)$ factor, one
  obtains a feasible fractional solution to the LP relaxation whose
  value is at least $w(S)/(1 + 3b \log m/\log \log m)$. This implies
  that $\lpopt \ge w(S)/(1 + 3b \log m/\log \log m)$, which proves the
  claim. Thus, if $\lpval \ge c \lpopt$ we have
  $E[w(S)] \ge (1-1/m^{\Omega(b)}) \lpval \ge c (1-1/m^{\Omega(b)})\lpopt$,
  and $w(S) \le (1+3b \log
  m/\log \log m) \lpopt$. Let $\alpha$ be the probability that
  $w(S) < (1-\eps)  \lpval$. Then we have the following.

  \begin{eqnarray*}
    E[w(S)] & \le & (1-\alpha) (1+3b \log m/\ln \ln m) \lpopt  + \alpha (1-\eps) \lpval  \\
    & \le & (1-\alpha) (1+3b \log m/\ln \ln m) \lpval/c + (1-\eps) \lpval
  \end{eqnarray*}
  However $E[w(S)] \ge (1-1/m^{\Omega(b)})\lpval$.   Rearranging and simplifying we have
  $$(1-\alpha) \ge c (\eps - m^{-\Omega(b)})/(1+3b \ln m/\ln \ln m).$$
  Note that $(1-\alpha)$ is the probability that $w(S) \ge (1-\eps)\lpval$. This finishes
  the proof.
\end{proof}

From the preceding lemma we see that the alteration algorithm
guarantees deterministically that the congestion is $O(b \ln m/\ln \ln
m)$ while the expected weight of the commodities routed by it is very
close to $\lpval$.  In fact with probability roughly
$\Omega(\eps/(1 + b \ln m))$ the value of the routed pairs is at least
$(1-\eps)\lpval$ assuming that $\lpval$ is a constant factor of $\lpopt$.
In most settings we would want to start with a fractional solution that is
very close to $\lpopt$ and hence the assumption is reasonable.
To guarantee a high probability bound\footnote{We say
that an event occurs with high probability, if it occurs with
probability at least $1-1/x^c$, where $x$ is the size of the input and
$c>0$ is a constant.} on achieving at least $(1-\eps)\lpval \geq
(1-\eps)\ipopt$, it suffices to repeat the algorithm $O(\ln m/\eps^2)$
times.

\section{Proofs of Chernoff Concentration Bounds}
\label{sec:concentration-bounds}

\allowdisplaybreaks

In this appendix we prove extended versions of the classic Chernoff bounds presented in Theorem~\ref{chernoff} to derandomize our approximation algorithm. The specific Chernoff extensions can be found in Appendix~\ref{app:chernoff-bound-proofs}, we first state some common convexity arguments below.

\subsection{Some Convexity Arguments}

\label{app:convexity-arguments}

\begin{lemma}
\label{app:lem:convexity-a}
The following holds for any $\param \in \mathbb{R}$ and $x \in [0,1]$:
\[
\exp(\param \cdot x) \leq 1 + (\exp(\param)-1)
\]
\end{lemma}

\begin{lemma}
\label{app:lem:convexity-b}
The following holds for any $\param \in \mathbb{R}$ and $x \in [0,1]$:
\[
1 + (\exp(\param)-1) \cdot x \leq \exp ((\exp(\param)-1)\cdot x)
\]
\end{lemma}

\begin{lemma}
\label{app:lem:convexity-rv}
Let $\rvX \in [0,1]$ denote a single random variable of expectation $\expX = \Expec[\rvX]$. For any $\param \in \mathbb{R}$ the following holds for the random variable $\rvY = \exp(\param \cdot \rvX)$:
\begin{align}
\Expec[\rvY] \leq \exp((\exp(\param)-1) \cdot \mu_i) \label{app:eq:convexity-rv}
\end{align}
\end{lemma}

\begin{lemma}
	\label{app:lem:chernoff-4-helper}
	The following inequality holds for any $x \in [0,1]$:
	\begin{align}
		(1+x)\ln(1+x)  -x \geq x^2/3\,. \label{app:eq:chernoff-4-helper}
	\end{align}
\end{lemma}

\begin{lemma}
\label{app:lem:chernoff-3-helper}
The following inequality holds for any $x \in [0,1)$:
\begin{align}
 -x -(1-x)\ln(1-x) \leq -x^2/2\,. \label{app:eq:chernoff-3-helper}
\end{align}
\end{lemma}

\begin{lemma}
	\label{app:lem:lnlnln}
	
	The following inequality holds for any $x > 0$:
	\begin{align}
		\ln \ln x - \ln \ln \ln x \geq 0.5\cdot\ln \ln x \,.
	\end{align}
\end{lemma}

\subsection{Chernoff Bounds}

In the following, Theorems~\ref{thm:chernoff-lower} and~\ref{thm:chernoff-upper} extend the classic Chernoff bounds of Theorem~\ref{chernoff}, enabling us to obtain pessimistic estimators for the sake of derandomization.

\label{app:chernoff-bound-proofs}

\begin{restatable}{theorem}{thmChernoffLowerBound}
\label{thm:chernoff-lower}
Let $\rvXSum$ be the sum of $\numberVars$ independent random variables $\rvX[1],\ldots,  \rvX[\numberVars]$ with $\rvX \in [0,1]$ for $\varEnum$.
Denoting by $\expXtilde \leq \expX = \Expec[X_\ell]$ lower bounds on the expected value of random variable $\rvX$, $\varEnum$, the following holds for any $\delta \in (0,1)$ with $\exptilde = \sumEnum[\expXtilde]$ and $\param = \ln(1-\delta)$;
\begin{align}
	\Prob[\rvXSum \leq (1-\delta) \cdot \exptilde] \overset{(a)}{\leq}
	\expT[-\param \cdot (1-\delta) \cdot \exptilde] \cdot {\prodEnum[\Expec[\expT[\param \cdot \rvX]]]} \overset{(b)}{\leq}  e^{-\delta^2\cdot \exptilde/2 } \label{app:eq:chernoff:lower}
\end{align}
\end{restatable}
\begin{proof}
We first prove the inequality $\Prob[\rvXSum \leq (1-\delta) \cdot \exptilde] \leq \left( \frac{e^{-\delta}}{{(1-\delta)^{1-\delta}}}\right)^{\exptilde}$ for $\delta \in (0,1)$. Let $\rvY = \exp(\param \cdot \rvX)$, $\varEnum$, for $\param = \ln(1-\delta)$. Note that $\param < 0$ holds.

By Lemma~\ref{app:lem:convexity-rv} Equality~\ref{app:eq:convexity-rv} holds.
As $\exp(\param)-1=-\delta<0$ holds for $\param<0$, the exponential function $f(z)=\exp((\exp(\param)-1)\cdot z)$ is monotonically decreasing. Using the lower bound  $\expXtilde \leq \expX$ and $\exptilde = \sumEnum[\expXtilde]$ the following is obtained:
\begin{align}
\Expec[\rvY] \leq \exp( (\exp(\param)-1) \cdot \expXtilde)
\label{app:eq:chernoff-3}
\end{align}

 As the variables $\rvX[1],\ldots, \rvX[\numberVars]$ are pairwise independent, the variables $\rvY[1],\ldots,\rvY[\numberVars]$ are also pairwise independent. Accordingly, the following holds for $\rvYSum=\expT[\param \cdot \rvXSum]$:
\begin{align}
\Expec[\rvYSum] = \Expec[\expT[\param\cdot \sumEnum[\rvX]]] = \Expec[\prodEnum[\expT[\param\cdot \rvX]]] =  \prodEnum[\Expec[\rvY]]\,.
\label{app:eq:chernoff-1}
\end{align}

Accordingly, the following is obtained:
\allowdisplaybreaks
\begin{alignat}{10}
& && \hspace{-20pt} \Pr[\rvXSum \leq (1 - \delta)\cdot \exptilde] \\
& = && \Pr[\expT[\param \cdot \rvXSum] \geq \expT[\param \cdot (1-\delta) \cdot \exptilde]]  &&&  \comm[as \ensuremath{\param < 0}]\\
& \leq && \frac{\Expec[\exp[\param \cdot \rvXSum]]}{\expT[\param \cdot (1-\delta) \cdot \exptilde]} &&& \comm[by Markov's inequality] \\
& = && \frac{\prodEnum[\Expec[\rvY]]}{\expT[\param \cdot (1-\delta) \cdot \exptilde]} &&& \comm[by Equation~\ref{app:eq:chernoff-1}] \label{app:eq:pessimistic-chernoff}\\
& \leq && \frac{\prodEnum[\expT[(\exp(\param)-1)\cdot \exptilde_i]]}{\expT[\param \cdot (1-\delta) \cdot \exptilde]} &&& \comm[by Equation~\ref{app:eq:chernoff-3}]\\
& = && \expT[\big( \sumEnum[ \expT[\param] - 1]\cdot \expXtilde \big) - \big( \param \cdot (1-\delta) \cdot \exptilde \big) ] \hspace{-30pt} &&& \comm[one exponent]\\
& = && \expT[\big( \sum_{i \in [N]} ( (1-\delta) - 1)\cdot \exptilde_i \big) - \big( \ln(1-\delta) \cdot (1-\delta) \cdot \exptilde \big)] &&& \comm[using \ensuremath{\param=\ln(1-\delta)}]\\
& = && \expT[\big( -\delta \cdot \exptilde \big) - \big( \ln(1-\delta) \cdot (1-\delta) \cdot \exptilde \big)] &&& \comm[definition of \ensuremath{\exptilde=\sum \nolimits_{\varEnum} \expXtilde}]\\[-8pt]
& = && \left( \frac{e^{-\delta}}{{(1-\delta)^{1-\delta}}}\right)^{\exptilde} &&& \label{app:eq:chernoff:result}
\end{alignat}

Given Equation~\ref{app:eq:chernoff:result}, Inequality~(b) is a corollary of Lemma~\ref{app:lem:chernoff-3-helper}, which showed the following for $\delta \in (0,1)$
\begin{align}
 -\delta -(1-\delta)\ln(1-\delta) \leq -\delta^2/2\,.
\end{align}
Multiplying both sides with $\exptilde$ and exponentiating both sides yields the desired result.

Regarding the Inequality~(a), we note that this follows by the above proof from Equation~\ref{app:eq:pessimistic-chernoff}.
\end{proof}

\begin{restatable}{theorem}{thmChernoffUpperBound}
	\label{thm:chernoff-upper}
	Let $\rvXSum$ be the sum of $\numberVars$ random variables $\rvX[1],\ldots,  \rvX[\numberVars]$ with $\rvX \in [0,1]$ for $\varEnum$.
	Denoting by  $\expXhat \geq \expX = \Expec[X_\ell]$ upper bounds on the expected value of random variable $\rvX$, $\varEnum$, the following holds for any $\delta > 0$ with $\exphat = \sumEnum[\expXhat]$ and $\param=\ln(\delta+1)$:
	\begin{align}
		\Prob[ \rvXSum \geq (1+\delta) \cdot \exphat] \overset{(a)}{\leq}
		  \expT[-\param \cdot (1+\delta) \cdot \exphat] \cdot \prodEnum[\Expec[\expT[\param \cdot \rvX]]]
		\overset{(b)}{\leq}
		\left( \frac{e^\delta}{{(1+\delta)^{1+\delta}}}\right)^{\exphat} \,. \label{app:eq:chernoff:upper}
	\end{align}
\end{restatable}
\begin{proof}
	Let $Y_i = \exp(\param \cdot X_i)$ for  $\param=\ln(\delta+1)$, $\varEnum$. Note that $\param > 0$ holds. As the variables $\rvX[1],\ldots,  \rvX[\numberVars]$ are pairwise independent, the variables $\rvY[1],\ldots,\rvY[\numberVars]$ are also pairwise independent. Considering $\rvY=\exp(\param \cdot \rvX)$, the following holds due to the pairwise independence of the variables $\rvY[1],\ldots,\rvY[\numberVars]$:
	\begin{alignat}{5}
		\textstyle
		\Expec[\rvY] & = &&\Expec[\exp(\param \cdot \rvXSum)] = \Expec[\expT[\param \cdot \sumEnum[\rvX]]] \notag\\
		&  = &&  \Expec[\prodEnum[\expT[\param \cdot \rvX]]] = \prodEnum[\Expec[Y_i]]\,. \label{app:eq:chernoff-6}
	\end{alignat}
	By Lemma~\ref{app:lem:convexity-rv} the Inequality~\ref{app:eq:convexity-rv} holds. As $\param > 0$ holds, the exponential function\break \mbox{$f(z) = \exp((\exp(\param)-1)\cdot z)$} is monotonically increasing. Using the upper bound  $\expXhat \geq \expX = \Expec[\rvX]$ on the expectation of $\rvX$, $\varEnum$, with $\exphat = \sumEnum[\expXhat] $, the following is obtained:
	\begin{align}
		\Expec[\rvY] \leq \exp( (\exp(\param)-1) \cdot \expXhat)\,.
		\label{app:eq:chernoff-2}
	\end{align}
	Using Equations~\ref{app:eq:chernoff-1} and~\ref{app:eq:chernoff-2} together with Markov's inequality, the left inequality of the  Chernoff bound \ref{app:eq:chernoff:upper} is obtained for any $\delta > 0$ by setting $\param = \ln(\delta + 1) > 0$:
	\allowdisplaybreaks
	\begin{alignat}{7}
		& && \hspace{-20pt} \Pr[X \geq (1+\delta)\cdot \exphat] \\
		& = && \Pr[\exp(\param \cdot X) \geq \exp(\param \cdot (1+\delta) \cdot \exphat )]  &&&  \\
		& \leq && \frac{\Expec[\exp(\param \cdot X)]}{\exp(\param \cdot (1+\delta) \cdot \exphat)} &&& \comm[by Markov's inequality] \\
		& = && \frac{\prod_{i \in [N]} \Expec[Y_i]}{\exp(\param \cdot (1+\delta) \cdot \exphat)} &&& \comm[by Equation~\ref{app:eq:chernoff-1}] \label{app:eq:pessimistic-chernoff-upper}\\
		& \leq && \frac{\prod_{i \in [N]} \exp((\exp(\param)-1)\cdot \exphat_i)}{\exp(\param \cdot (1+\delta) \cdot \exphat)} &&& \comm[by Equation~\ref{app:eq:chernoff-2}]  \\
		& = && \exp \bigg( \big( \sum_{i \in [N]} ( \exp(\param) - 1)\cdot \exphat_i \big) - \big( \param \cdot (1+\delta) \cdot \exphat \big) \bigg) \hspace{-30pt} &&& \\
		& = && \exp \bigg( \big( \sum_{i \in [N]} ( (1+\delta) - 1)\cdot \exphat_i \big) - \big( \ln(1+\delta) \cdot (1+\delta) \cdot \exphat \big) \bigg) &&& \comm[using $\param=\ln(\delta+1)$]\\
		& = && \exp \bigg( \big( \delta \cdot \exphat \big) - \big( \ln(1+\delta) \cdot (1+\delta) \cdot \exphat \big) \bigg) &&& \comm[def. of $ \exphat=\sum_{i \in [N]} \exphat_i$]\\
		& = && \frac{\exp \bigg( \delta \cdot \exphat \bigg)}{\exp \bigg( \ln(1+\delta)\cdot (1+\delta)\cdot \exphat \bigg) } = \left( \frac{e^\delta}{{(1+\delta)^{1+\delta}}}\right)^{\exphat} &&&
	\end{alignat}
	This completes the proof of inequality (b). Regarding the Inequality~(a), we note that this follows by the above proof from Equation~\ref{app:eq:pessimistic-chernoff-upper}.
\end{proof}

\end{document}